\documentclass[12pt]{article}

\usepackage{epsfig,amsmath,amssymb,amsfonts,amstext,amsthm,mathrsfs,cleveref}
\usepackage{latexsym,graphics,epsf,epsfig,psfrag}
\usepackage{cite}
\usepackage{color}
\usepackage{epstopdf}

\usepackage{float}
\usepackage{graphicx}

\newtheorem{corollary}{\textbf{Corollary}}

\newtheorem{theorem}{\textbf{Theorem}}
\newtheorem{proposition}{\textbf{Proposition}}
\newtheorem{remark}{\textbf{Remark}}

\newcommand{\nn}{\nonumber}

\newcommand{\cX}{\mathcal{X}}

\newcommand{\cW}{\mathcal{W}}

\newcommand{\tl}{\tilde{l}}

\newcommand{\tR}{\tilde{R}}

\newcommand{\hl}{\hat{l}}

\newcommand{\hw}{\hat{w}}

\DeclareMathAlphabet{\matheuf}{U}{euf}{m}{n}

\begin{document}
%

  \begin{center}
  	\baselineskip 1.3ex {\Large \bf State-Dependent Interference Channel with Correlated States
  		\footnote{The work of Y. Sun and Y. Liang was supported by the National Science Foundation under Grant CCF-1618127. The work of S. Shamai was supported by the European Union's Horizon 2020 Research And Innovation Programme, grant agreement No. 694630.}
  		\\
  	}
  	\vspace{0.15in} Yunhao Sun, \footnote{Yunhao Sun is with the Department of Electrical Engineering and Computer Science, Syracuse University, Syracuse, NY 13244 USA (email: ysun33@syr.edu).} Ruchen Duan,
  	\footnote{Ruchen Duan is with Samsung Semiconductor Inc., San Diego, CA 92121 USA (email: r.duan@samsung.com).} Yingbin Liang, \footnote{Yingbin Liang is with the Department of Electrical and Computer Engineering, The Ohio State University, Columbus, OH 43210 USA (email: liang.889@osu.edu).}
  	Shlomo Shamai (Shitz)\footnote{Shlomo Shamai (Shitz) is with the Department of Electrical Engineering, Technion-Israel Institute of Technology, Technion city, Haifa 32000, Israel (email: sshlomo@ee.technion.ac.il).}
  	
  \end{center}

%



\begin{abstract}
	This paper investigates the Gaussian state-dependent interference channel (IC) and Z-IC, in which two receivers are corrupted respectively by two different but correlated states that are noncausally known to two transmitters but are unknown to the receivers. Three interference regimes are studied, and the capacity region boundary or the sum capacity boundary is characterized either fully or partially under various channel parameters. In particular, the impact of the correlation between states on cancellation of state and interference as well as achievability of capacity is explored with numerical illustrations. For the very strong interference regime, the capacity region is achieved by the scheme where the two transmitters implement a cooperative dirty paper coding. For the strong but not very strong interference regime, the sum-rate capacity is characterized by rate splitting, layered dirty paper coding and successive cancellation. For the weak interference regime, the sum-rate capacity is achieved via dirty paper coding individually at two receivers as well as treating interference as noise. 
\end{abstract}

\section{Introduction}

State-dependent interference channels (ICs) are of great interest in wireless communications, in which receivers are interfered not only by other transmitters' signals but also by independent and identically distributed (i.i.d.) state sequences. The state can capture interference signals that are informed to transmitters, and are hence often assumed to be noncausally known by these transmitters in the model. Such interference cognition can occur in practical wireless networks due to node coordination or backhaul networks. 

Both the state-dependent IC and Z-IC have been studied in the literature. The state-dependent IC was studied in \cite{Zhang11a,Zhang11b} with two receivers corrupted by the same state, and in \cite{Ghas13ISWCS} with two receivers corrupted by independent states. In \cite{Somekh08,Duan15IT}, two state-dependent cognitive IC models were studied, where one transmitter knows both messages, and the two receivers are corrupted by two states. More recently, in\cite{Duan16IT}, both the state-dependent regular IC and Z-IC were studied, where the receivers are corrupted by the same but differently scaled state. Furthermore, in \cite{Duan13ITW,Ghas14}, a type of the state-dependent Z-IC was studied, in which only one receiver is corrupted by the state and the state information is known only to the other transmitter. In\cite{Kolte14}, a class of deterministic
state-dependent Z-ICs was studied, where two receivers are corrupted by the same state and the state information is known only to one transmitter. In \cite{Fehri2015Z}, a type of the state-dependent Z-IC with two states was studied, where each transmitter knows only the state that corrupts its corresponding receiver. In \cite{Haji13}, a state-dependent Z-interference broadcast channel was studied, in which one transmitter has only one message for its corresponding receiver, and the other transmitter has two messages respectively for two receivers. Both receivers are corrupted by the same state, which is known to both transmitters.

In all the previous work of the state-dependent IC and Z-IC, the states at two receivers are either assumed to be independent, or to be the same but differently scaled, with the exception of \cite{Fehri2015Z} that allows correlation between states. However, \cite{Fehri2015Z} assumes that each transmitter knows only one state at its corresponding receiver, and hence two transmitters cannot cooperate to cancel the states. In this paper, we investigate the state-dependent IC and Z-IC with the two receivers being corrupted respectively by two correlated states and with both transmitters knowing both states in order for them to cooperate. The main focus of this paper is on the Gaussian state-dependent IC and Z-IC, where the receivers are corrupted by additive interference, state, and noise. The aim is to design encoding and decoding schemes to handle interference as well as to cancel the state at the receivers. In particular, we are interested in answering the following two questions: (1) whether and under what conditions both states can be simultaneously fully canceled so that the capacity for the IC and Z-IC without state can be achieved; and (2) what is the impact of the correlation between two states on state cancellation and capacity achievability. 

We summarize our results as follows. Our novelty of designing achievable schemes lies in joint design of the interference cancellation schemes together with the Gel'fand-Pinsker binning \cite{Gelf80} and dirty paper coding \cite{Costa83} for state cancellation in order to characterize the capacity region. More specifically, we study three interference regimes. For the very strong interference regime, we characterize the channel parameters under which the two receivers achieve their corresponding point-to-point channel capacity without state and interference. Thus, the interference as well as states are fully canceled, and the capacity region is characterized as a rectangular region. In particular, we demonstrate the impact of the correlation between the two states in such a regime. Interestingly, we show that high interference may not always be beneficial for canceling both state and interference, which is in contrast to the IC without state. For the strong interference regime, we characterize the sum capacity boundary partially under certain channel parameters based on joint design of rate splitting, successive cancellation, as well as dirty paper coding. We also explain how the correlation affects the achievability of the sum capacity, and affects the comparison between the IC and Z-IC. For the weak interference regime, we observe that the sum capacity can be achieved by the two transmitters independently employing dirty paper coding and receiver 1 treating interference as noise as shown in \cite{Duan16IT} for the same but differently scaled state at the two receivers. The sum capacity is not affected by the correlation between states. 


\section {Channel Model}\label{sec:vs_model}

\begin{figure}[thb]
	\centering
	\includegraphics[width=3.5in]{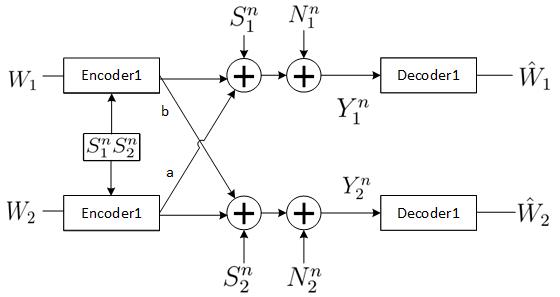}\label{fig:channel}
	\caption{The state-dependent IC}
\end{figure}

We consider the state-dependent IC (as shown in Fig.~\ref{fig:channel}), in which transmitters 1 and 2 send messages $W_{1}$ and $W_{2}$ respectively to the receivers 1 and 2. For $k=1,2$, encoder $k$ maps the message $w_k\in \cW_k$ to a codeword $x_k^n\in \cX_k^n$. The two inputs $x_1^n$ and $x_2^n$ are then transmitted over the IC to the receivers, which are corrupted by two correlated state sequences $S_1^n$ and $S_2^n$, respectively. The state sequences are known to both the transmitters noncausally, but are unknown at the receivers. Encoders 1 and 2 want to map their messages as well as the state sequences' information into codewords $x_1^n\in\cX_1^n$ and $x_2^n\in\cX_2^n$. The channel transition probability is given by $P_{Y_1Y_2|S_1S_2X_1X_2}$. The decoders at the receivers map the received sequences $y_1^n$ and $y_2^n$ into corresponding messages $\hw_k\in \cW_k$ for $k=1,2$.

The average probability of error for a length-$n$ code is defined as
\begin{flalign}\label{PE}
P_e^{(n)} = & \frac{1}{|\cW_1||\cW_2|}\sum_{w_1=1}^{|\cW_1|}\sum_{w_2=1}^{|\cW_2|} Pr\lbrace(\hat{w}_1, \hat{w}_2) \neq (w_1, w_2)\rbrace.
\end{flalign}
A rate pair $(R_1, R_2)$ is {\em achievable} if there exist a sequence of message sets $\cW_{k}^{(n)}$ with $|\cW_{k}^{(n)}|=2^{nR_k}$ for $k=1, 2$, such that the average error probability $P_e^{(n)} \rightarrow 0$ as $n \to \infty$. The {\em capacity region} is defined to be the closure of the set of all achievable rate pairs $(R_1, R_2)$.

In this paper, we study the Gaussian channel with the outputs at the two receivers for one channel use given by
\begin{subequations}
	\begin{flalign}
	Y_1&=X_1+ aX_2+S_1+N_1,\label{eq:GeneralChannelModel}\\
	Y_2&=bX_1+X_2+S_2+N_2
	\end{flalign}
\end{subequations}	
where $a$ and $b$ are the channel gain coefficients, and $N_1$ and $N_2$ are noise variables with Gaussian distributions $N_1 \sim \mathcal{N}(0,1)$ and $N_2 \sim \mathcal{N}(0,1)$. The state variables $S_1$ and $S_2$ are jointly Gaussian with the correlation coefficient $\rho$ and the marginal distributions $S_1 \sim \mathcal{N}(0,Q_1)$ and $S_2\sim \mathcal{N}(0,Q_2)$. Both the noise variables and the state variables are i.i.d. over the channel uses. The channel inputs $X_1$ and $X_2$ are subject to the average power constraints $P_1$ and $P_2$. 

The Z-IC is also studied in this paper, in which receiver 2 is not corrupted by the interference $X_1$ from transmitter 1 (i.e, $b$=0 for Gaussian channel).

Our goal is to characterize channel parameters, under which the capacity of the corresponding IC and Z-IC without the presence of the state can be achieved, and thus the capacity region of the IC and Z-IC with the presence of state is also established. In particular, we are interested in understanding the impact of the correlation between the states $S_1$ and $S_2$ on the capacity characterization.

\section {Very Strong Interference Regime}\label{sec:vs_regime}

In this section, we study the impact of the correlation between states on the characterization of the capacity in the very strong interference regime. We study both the state-dependent IC and Z-IC.

\subsection{State-Dependent IC}\label{vs_regular}
In this subsection, we study the state-dependent IC in the very strong interference regime, where the channel parameters satisfy
\begin{subequations}\nn
	\begin{flalign}
		P_1+a^2P_2+1&>(1+P_1)(1+P_2),\\
		b^2P_1+P_2+1&>(1+P_1)(1+P_2).
	\end{flalign}
\end{subequations}	
For the corresponding IC without states, the capacity region contains rate pairs ($R_1,R_2$) satisfying:
\begin{equation}\label{cap:VeryStrong}
	\begin{aligned}
		R_1 \leqslant & \frac{1}{2}\log(1+P_1),\\
		R_2 \leqslant & \frac{1}{2}\log(1+P_2).
	\end{aligned}
\end{equation}	
In this case, the two receivers achieve the point-to-point channel capacity without interference. Furthermore, in \cite{Duan16IT}, an achievable scheme has been established to achieve the same point-to-point channel capacity when the two receivers are corrupted by the same but differently scaled state. Our focus here is on the more general scenario, where the two receivers are corrupted by two {\em correlated} states, and our aim is to understand how the correlation affects the design of the scheme.


We first design an achievable scheme to obtain an achievable rate region for the discrete memoryless IC. The two transmitters encode their messages $W_1$ and $W_2$ into two auxiliary random variables $U$ and $V$, respectively, based on the Gel'fand-Pinsker binning scheme. Since the channel satisfies the very strong interference condition, it is easier for receivers to decode the information of the interference. Thus each receiver first decodes the auxiliary random variable corresponding to the message intended for the other receiver, and then decodes its own message by decoding the auxiliary random variable for itself. For instance, receiver 1 first decodes $V$, then uses it to cancel the interference $X_2$ and partial state interference, and finally decodes its own message $W_1$ by decoding $U$. Differently from \cite{Duan16IT}, two auxiliary random variables $U$ and $V$ are designed not with regard to one state, but with regard to two correlated states. This requires a joint design for $U$ and $V$ to fully cancel the states. Based on such a scheme, we obtain the following achievable region.

\begin{proposition}\label{pps:IC inner}
	For the state-dependent IC with states noncausally known at both transmitters, the achievable region consists of rate pairs $(R_1,R_2)$ satisfying:
	\begin{subequations}
		\begin{flalign}
		R_1 \leqslant & \min\{ I(U;VY_1),I(U;Y_2)\}-I(S_1S_2;U),\label{eq:pps1-1} \\
		R_2 \leqslant & \min\{ I(V;UY_2),I(V;Y_1)\}-I(S_1S_2;V)\label{eq:pps1-2}
		\end{flalign}
	\end{subequations}
	for some distribution $P_{S_1S_2}P_{U|S_1S_2}P_{X_1|US_1S_2}P_{V|S_1S_2}P_{X_2|VS_1S_2}P_{Y_1Y_2|S_1S_2X_1X_2}$, where $U$ and $V$ are auxiliary random variables.
\end{proposition}
\begin{proof}
	See Appendix \ref{apx:IC inner}.
\end{proof}

We now study the Gaussian IC. For the sake of technical convenience, we express the Gaussian channel in Section \ref{sec:vs_model} in a different form. Since $S_1$ and $S_2$ are jointly Gaussian, $S_1$ can be expressed as $S_1=dS_2+S_1^\prime$ where $d$ is a constant representing the level of correlation, and $S_1'$ is independent from $S_2$ and $S_1'\sim \mathcal{N}(0, Q_1')$ with $Q_1=d^2Q_2+Q_1^\prime$. Thus, without loss of generality, the channel model can be expressed in the following equivalent form that is more convenient for analysis.
\begin{subequations}
	\begin{flalign}
		Y_1&=X_1+ aX_2+dS_2+S_1^\prime+N_1,\\
		Y_2&=bX_1+X_2+S_2+N_2.
	\end{flalign}
\end{subequations}	

Following Proposition \ref{pps:IC inner}, we characterize the condition under which both the state and interference can be fully canceled, and hence the capacity region for the state-dependent Gaussian IC in the very strong interference regime is obtained.

\begin{theorem}\label{thr:IC inner}
	For the state-dependent Gaussian IC with state noncausally known at both transmitters, the capacity region is the same as the point-to-point channel capacity for both receivers, if the channel parameters  $(a,b,d,P_1,P_2,Q_1^\prime,Q_2)$ satisfy the following conditions:
	\begin{subequations}
		\begin{flalign}
		\frac{1}{2}\log(1+P_1)\leqslant& h(X_1)-h(U,Y_2)+h(Y_2),\label{eq:cond1}\\
		\frac{1}{2}\log(1+P_2)\leqslant& h(X_2)-h(V,Y_1)+h(Y_1),\label{eq:cond2}
		\end{flalign}
	\end{subequations}
where the auxiliary random variables are designed as $U=X_1+\alpha_1S_1^\prime+\alpha_2S_2$ and $V=X_2+\beta_1S_1^\prime+\beta_2S_2$. Here, $X_1$,$X_2$, $ S_1^\prime $ and $ S_2 $ are independent Gaussian variables with mean zero and variances $P_1$,$P_2$, $ Q_1 $ and $ Q_2 $, respectively. The parameters $\alpha_1$,$\alpha_2$,$\beta_1$ and $\beta_2$ are set as
\begin{equation}\label{eq:s_variablesetting}
	\begin{aligned}
		\alpha_1&=\frac{P_1(1+P_2)}{(P_1+1)(P_2+1)-abP_1P_2},\ \ \ \  \alpha_2=\frac{P_1(d+dP_2-aP_2)}{(P_1+1)(P_2+1)-abP_1P_2},\\
		\beta_1&=\frac{bP_1P_2}{(P_1+1)(P_2+1)-abP_1P_2},\ \ \ \  \beta_2=\frac{P_2(P_1+1-bdP_1)}{(P_1+1)(P_2+1)-abP_1P_2}.
	\end{aligned}	
\end{equation}
	\end {theorem}
	
	\begin{proof}
		 The proof mainly follows Proposition \ref{pps:IC inner}. We design $U$ and $V$ as stated in Theorem \ref{thr:IC inner}.  As discussed in the proof of Proposition \ref{pps:IC inner}, $V$ is first decoded by receiver 1 and $U$ is first decoded by receiver 2. And then  receiver 2 subtracts $U$ to cancel $X_1$ and obtain  $Y_2^{\prime}=Y_2-bU=X_2-b\alpha_1S_1^\prime+(1-b\alpha_2)S_2+N_2$, and receiver 1 subtracts $V$ to cancel $X_2$ and obtain  $Y_1^{\prime}=Y_1-aV=X_1+(1-a\beta_1)S_1^\prime+(d-a\beta_2)S_2+N_1$. In order to fully cancel the channel states for $Y_1^{\prime}$ and $Y_2^{\prime}$, based on the dirty paper coding scheme, we further require the coefficients to satisfy the following conditions,
		\begin{subequations}
			\begin{flalign}
			\frac{\alpha_1}{1-a\beta_1}&=\frac{\alpha_2}{d-a\beta_2}\label{eq:dirty_con1}\\
			\frac{\alpha_1}{1-a\beta_1}&=\frac{P_1}{P_1+1}\label{eq:dirty_con2}\\
			\frac{\beta_1}{-b\alpha_1}&=\frac{\beta_2}{1-b\alpha_2}\label{eq:dirty_con3}\\
			\frac{\beta_1}{-b\alpha_1}&=\frac{P_2}{P_2+1}\label{eq:dirty_con4}
			\end{flalign}
		\end{subequations}		
		which yield $\alpha_1$,$\alpha_2$,$\beta_1$ and $\beta_2$ in \eqref{eq:s_variablesetting}.
		
		By plugging these parameters into \eqref{eq:pps1-1}, we obtain $$I(U;VY_1)-I(S_1,S_2;U)= \frac{1}{2}\log(1+P_1),$$
		which yields $$R_1 \leqslant \min\{I(U;Y_2)-I(S_1,S_2;U), \frac{1}{2}\log(1+P_1)\}.$$

		Similarly, \eqref{eq:pps1-2} yields $$R_2 \leqslant \min\{I(V;Y_1)-I(S_1,S_2;V), \frac{1}{2}\log(1+P_2)\}.$$
	In order to achieve the channel capacity of the point-to-point channel as shown in \eqref{cap:VeryStrong} for both receivers, the following conditions should be satisfied:
	\begin{subequations}
		\begin{flalign}
		\frac{1}{2}\log(1+P_1)&\leqslant I(U;Y_2)-I(S_1,S_2;U)\label{eq:thr1_con1}\\
		\frac{1}{2}\log(1+P_2)&\leqslant I(V;Y_1)-I(S_1,S_2;V) \label{eq:thr1_con2},
		\end{flalign}
	\end{subequations}	
which completes the proof.
	\end{proof}	
	We note that the conditions in Theorem \ref{thr:IC inner} represent the comparison between the ability of receivers to decode messages in different decoding steps. For instance, in condition \eqref{eq:cond1} the right-hand side term represents how much receiver 2 can decode $U$ in the first step of decoding in order to cancel the interference, and the left-hand side term represents the rate at which receiver 1 can decode $U$ in the second step of decoding, where we can use the dirty paper coding scheme to fully cancel the states and achieve the capacity. Hence, achieving the point-to-point channel capacity requires the second step to dominate the performance.
		\begin{figure}[thb]
		\centering
		\begin{tabular}{ccc} 
			\includegraphics[width=2in]{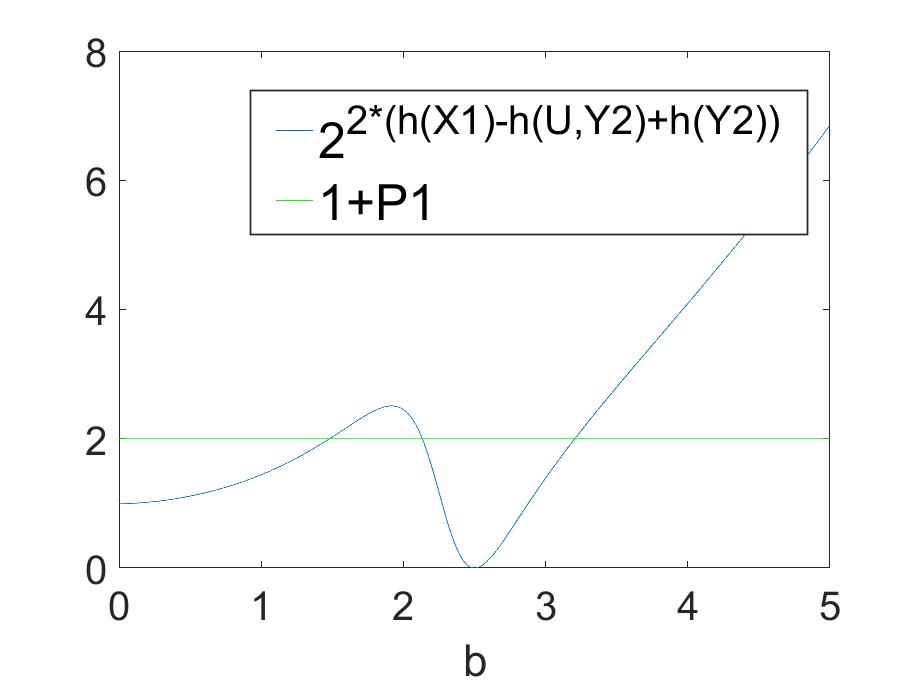}
			&\includegraphics[width=2in]{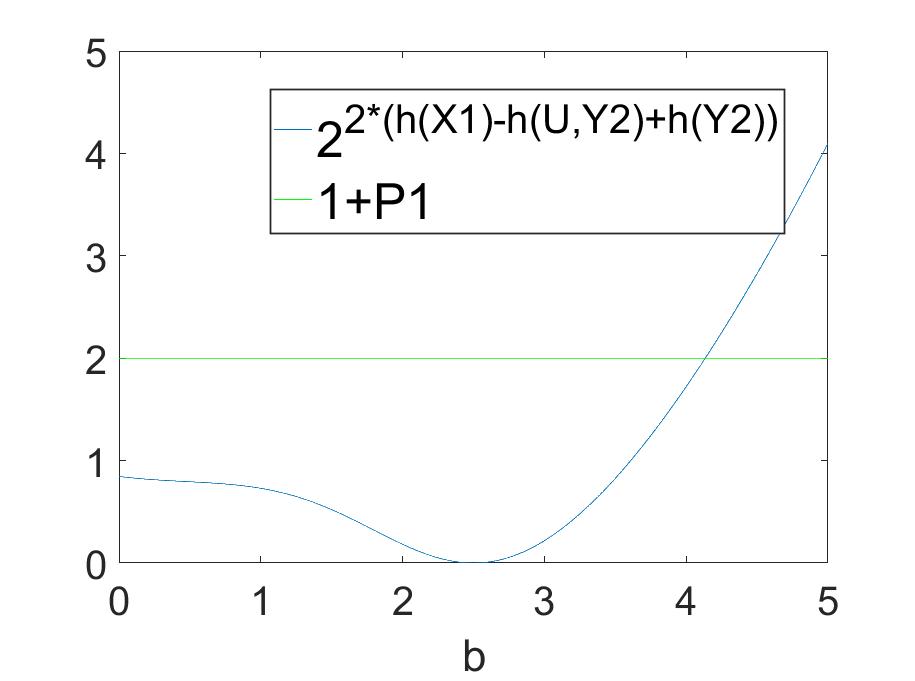}
			&\includegraphics[width=2in]{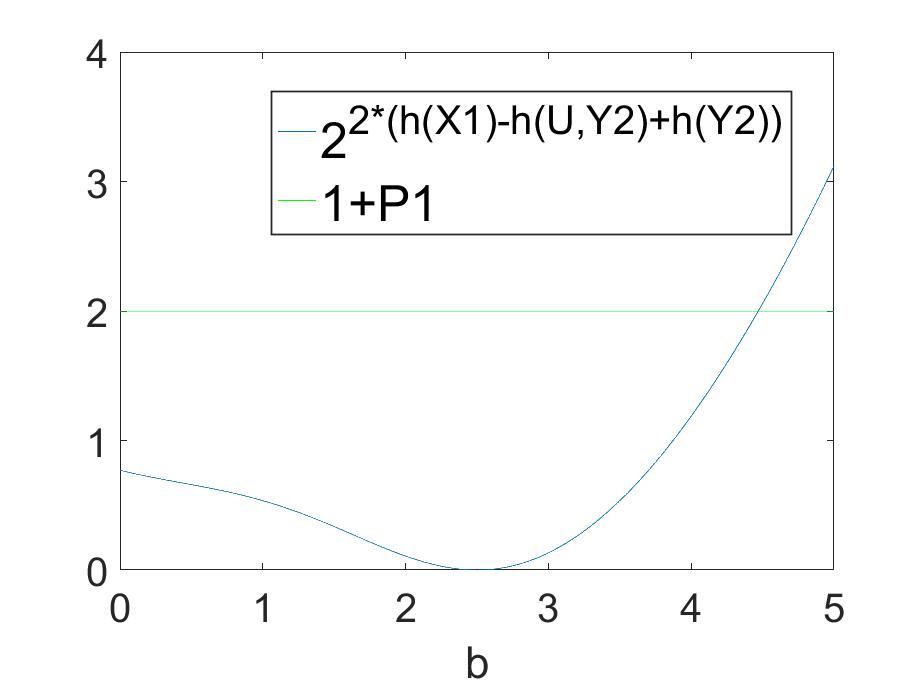}\\
			$d=0.99$&$d=0.5$&$d=0.1$\\
		\end{tabular}
		\begin{tabular}{ccc} 
			\includegraphics[width=2in]{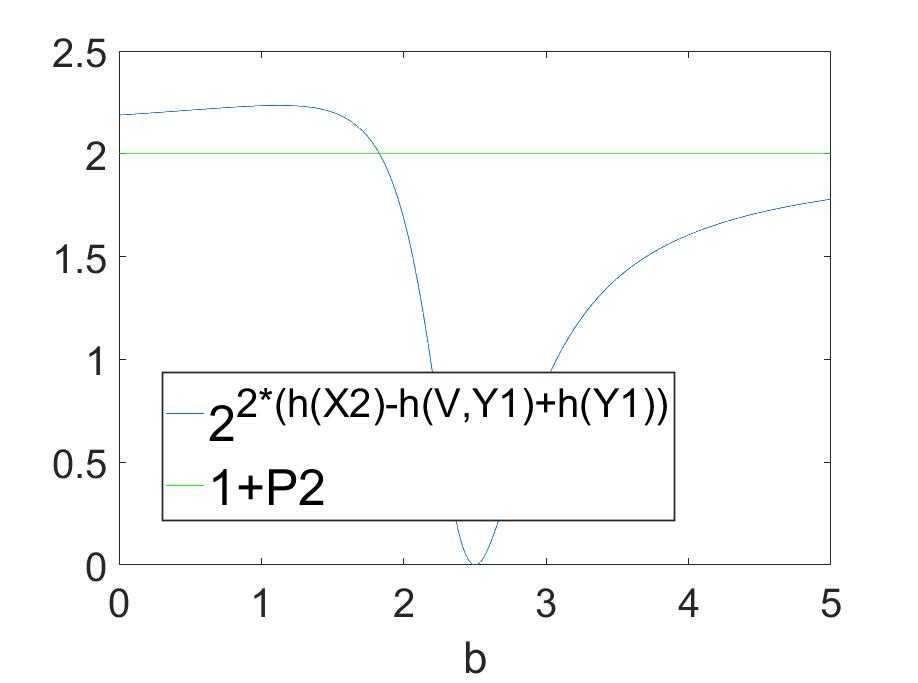}
			&\includegraphics[width=2in]{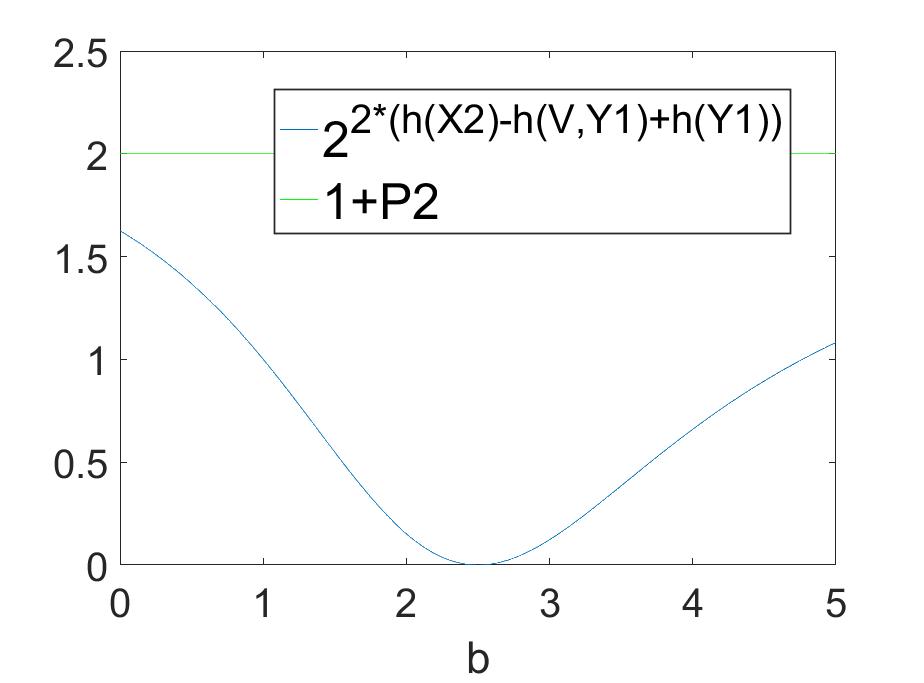}
			&\includegraphics[width=2in]{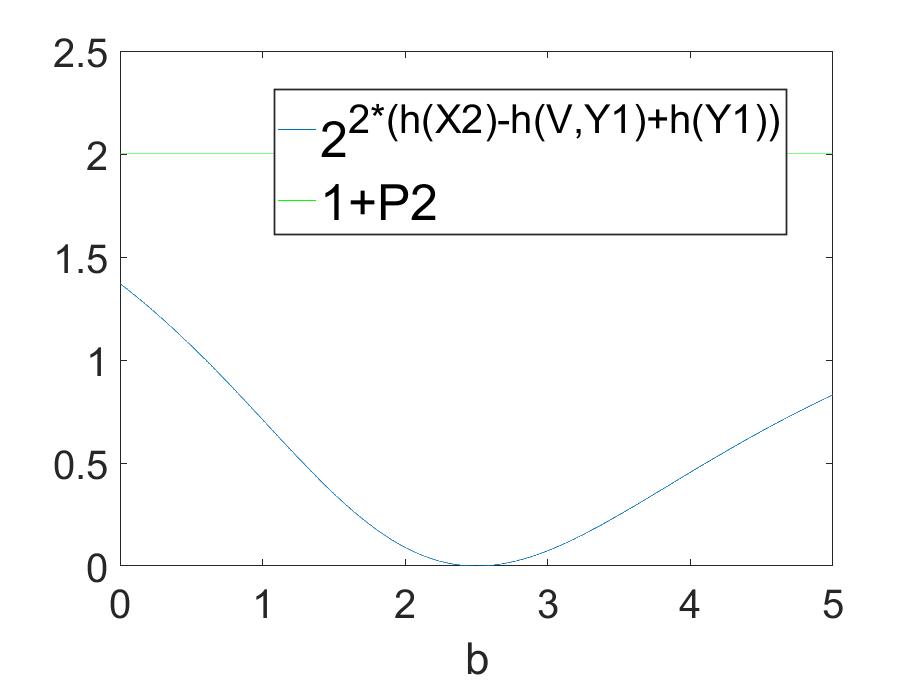}\\
			$d=0.99$&$d=0.5$&$d=0.1$\\
		\end{tabular}
		\caption{Comparison of values of both sides of \eqref{eq:cond1_1} and \eqref{eq:cond2_1}}\label{fig:b_beta}
	\end{figure}

	We next study the impact of the channel parameters and state correlation on the achievablility of the point-to-point capacity. In particular, we illustrate how the interference gains $(a,b)$ affect the conditions (\ref{eq:cond1}) and (\ref{eq:cond2}). To make the figure more clear, we take exponential of both sides, and hence the conditions (\ref{eq:cond1}) and  (\ref{eq:cond2}) become:
	\begin{subequations}
		\begin{flalign}
		1+P_1\leqslant& 2^{2(h(X_1)-h(U,Y_2)+h(Y_2))},\label{eq:cond1_1}\\
		1+P_2\leqslant& 2^{2(h(X_2)-h(V,Y_1)+h(Y_1))}.\label{eq:cond2_1}
		\end{flalign}
	\end{subequations}	
%
In Fig. \ref{fig:b_beta}, we set $Q_1=Q_2=0.9, P_1=1,P_2=1$ and $a=1.6$, and plot the change of both left-hand side terms and right-hand side terms in  (\ref{eq:cond1_1}) and  (\ref{eq:cond2_1}) versus the channel parameters $ b $ for three different values of $ d $. Taking the first row of Fig. \ref{fig:b_beta} as an example, it is clear that $1+P_1$ is a straight line, and $2^{2(h(X_1)-h(U,Y_2)+h(Y_2))}$ is not a monotone function with respect to $ b $. The condition \eqref{eq:cond1_1} (and hence \eqref{eq:cond1}) is satisfied only when $2^{2(h(X_1)-h(U,Y_2)+h(Y_2))}$ is above the straight line  $1+P_1$. When the parameter $d=0.99$, there are two regions over which the condition (\ref{eq:cond1_1}) is satisfied. But if $d=0.5$, there is only one region over which the condition (\ref{eq:cond1_1}) is satisfied. For $d=0.1$, there is also only one region where the condition (\ref{eq:cond1_1}) is satisfied. Similarly, the second row in Fig. \ref{fig:b_beta} illustrates the regions of $b$ over which the condition (\ref{eq:cond2_1}) is satisfied for the corresponding values of $ d $. Then the intersection of the region of $ b $ in the first and second rows of Fig. \ref{fig:b_beta} fully determines the ranges of $ b $ over which the point-to-point channel capacity can be achieved for both receivers.
		
\begin{figure}[thb]
	
	\begin{tabular}{ccc} 
		\includegraphics[width=2in]{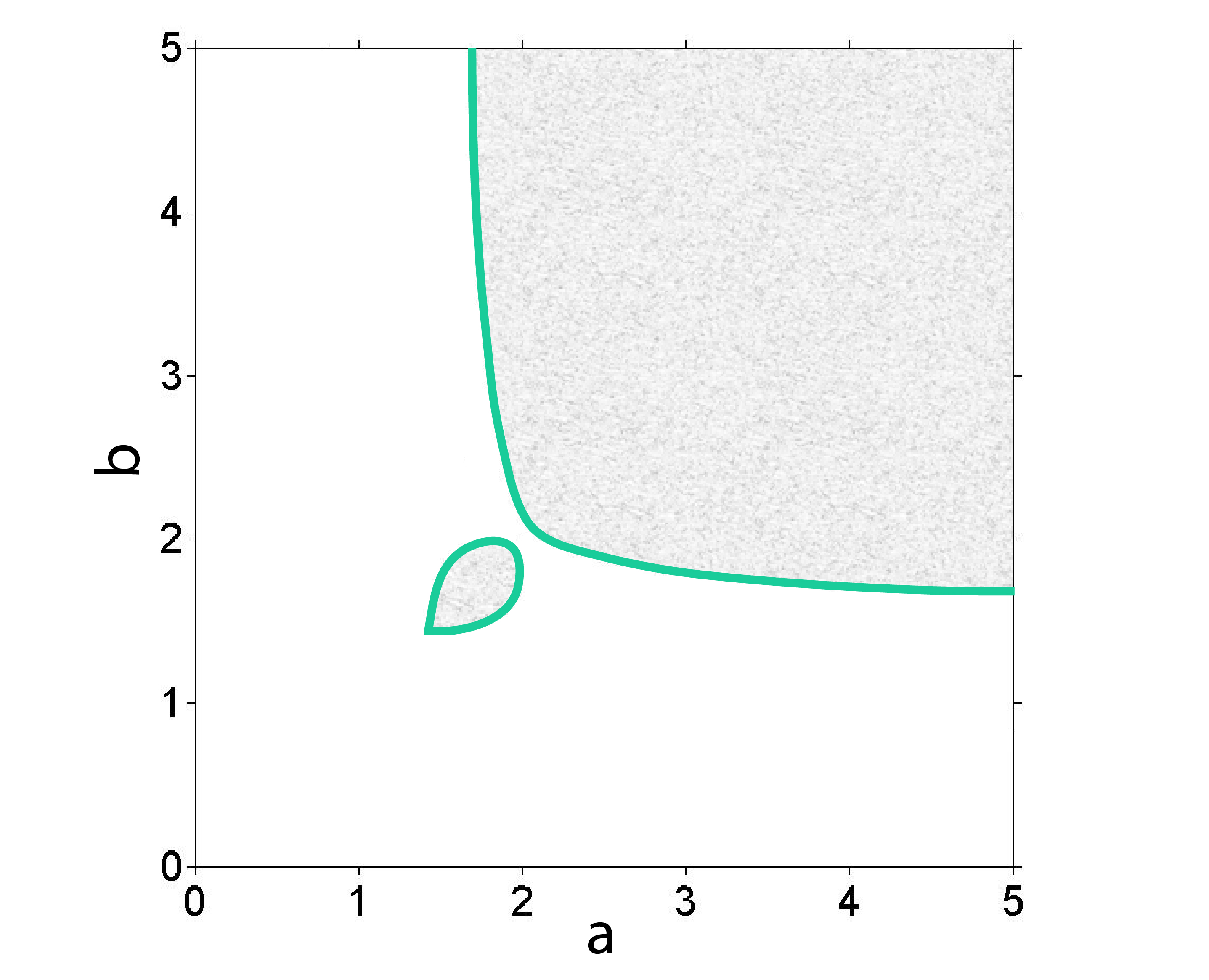}
		&\includegraphics[width=2in]{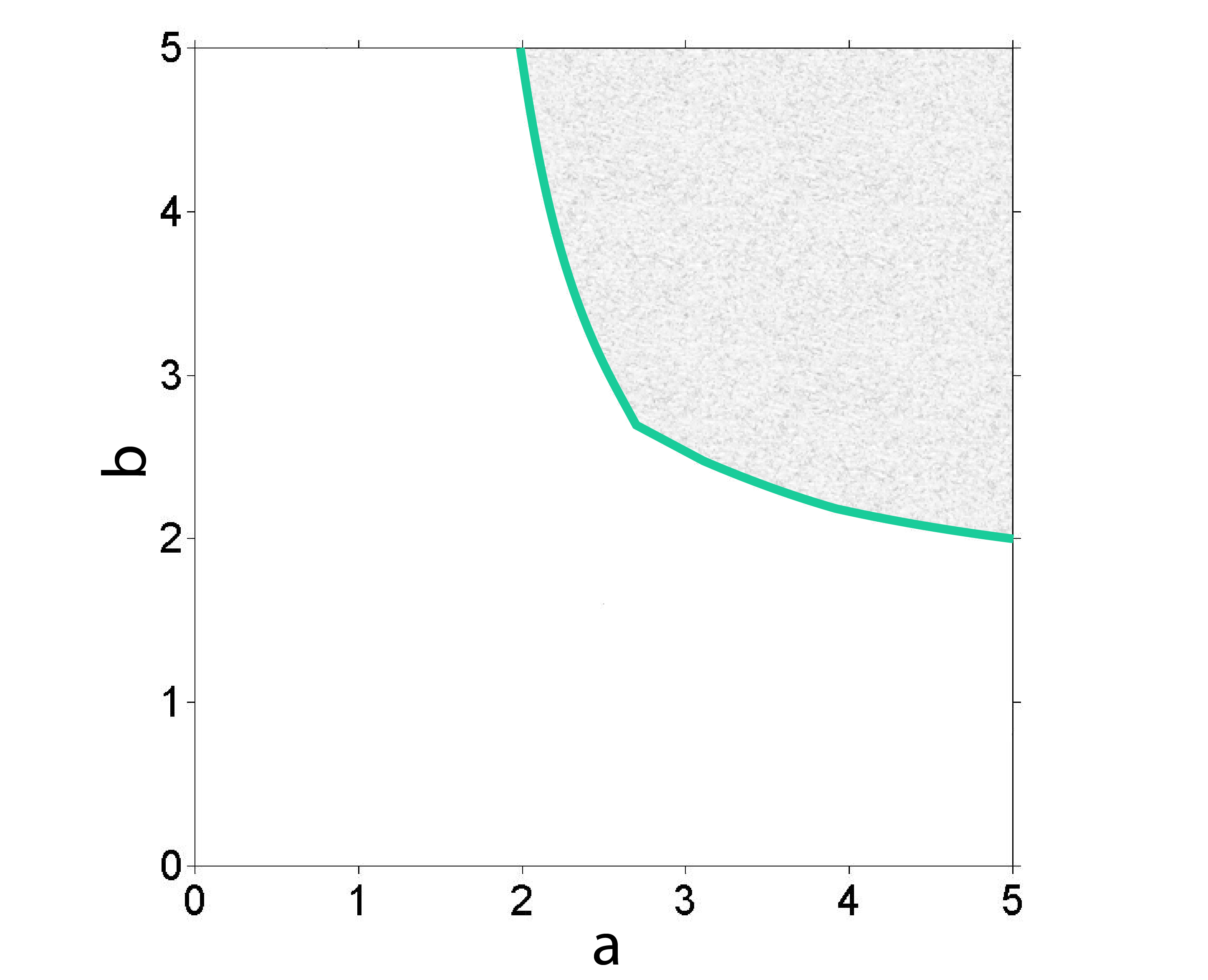}
		&\includegraphics[width=2in]{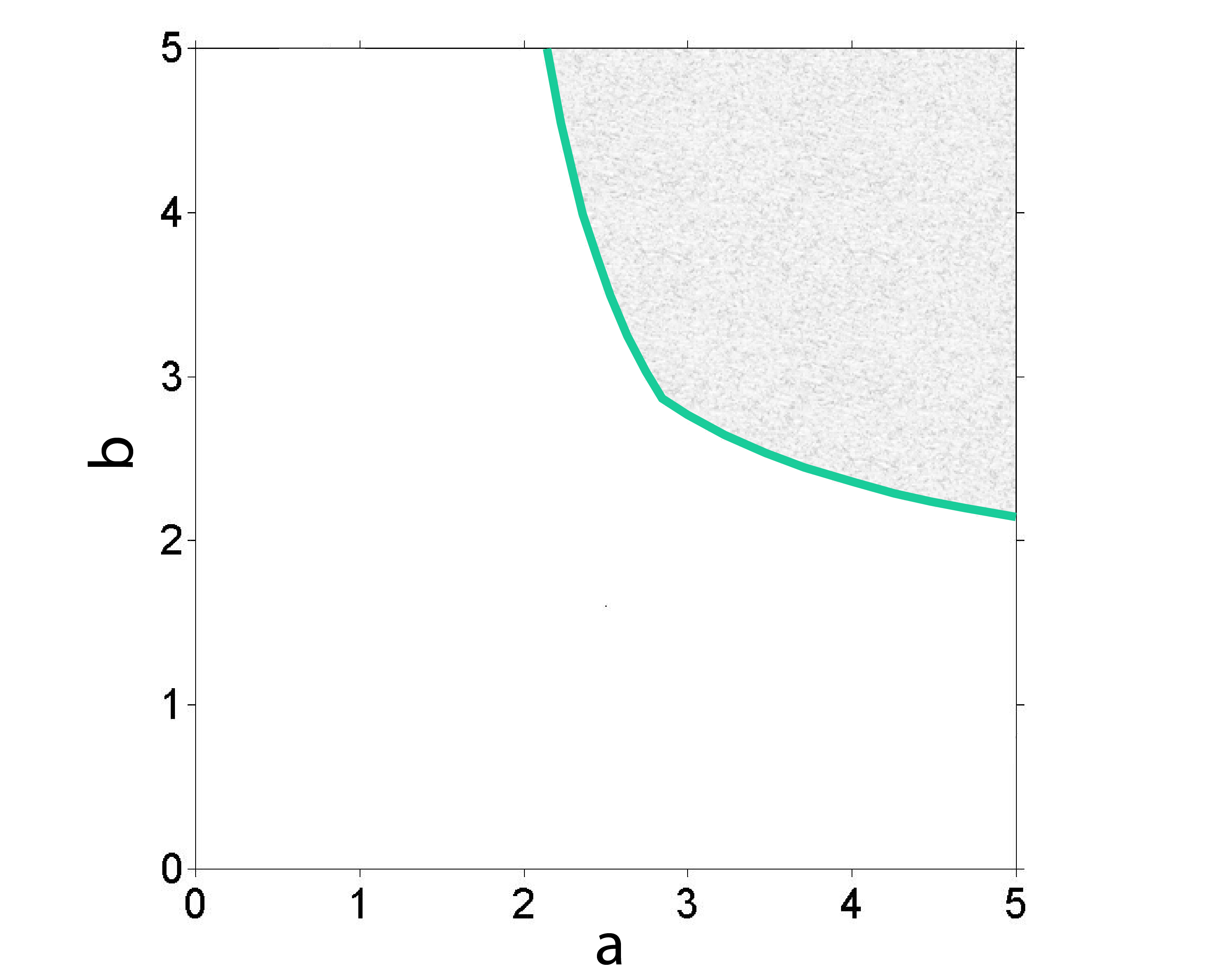}\\
		$d=0.99$&$d=0.5$ &$d=0.1$
	\end{tabular}
	\caption{Ranges of (a,b) under which the point-to-point channel capacity is achieved.}\label{fig:c}
\end{figure}
	
The range of the parameters $(a,b)$ such that the point-to-point channel capacity is obtained is shown in Fig. \ref{fig:c}. For these figures, if we fix $a=1.6$, then the ranges of b is consistent with those in Fig. \ref{fig:c} where both (\ref{eq:cond1_1}) and (\ref{eq:cond2_1}) are satisfied. Fig. \ref{fig:c} also illustrates the impact of the correlation $d$ between the states $S_1$ and $S_2$ on the achievability of channel capacity. It is clear that as $d$ increases, i.e., the two states are more correlated, the range of $(a,b)$ over which the point-to-point channel capacity is achieved gets larger. This confirms the intuition that more correlated states are easier to be fully canceled.

\subsection{State-Dependent Z-IC}\label{vs_zic}
In this subsection, we study the state-dependent Z-IC in the very strong interference regime, in which the channel parameters satisfy $a^2>1+P_1$. Similarly to the state-dependent IC, we study the ranges of the channel parameters, over which the capacity is also characterized by \eqref{cap:VeryStrong}. We note that here Z-IC (with $b=0$) cannot be viewed as a special case of the IC in the very strong interference regime.

We first design an achievable scheme to obtain an achievable rate region for the discrete memoryless Z-IC. The two transmitters encode their messages $W_1$ and $W_2$ into two auxiliary random variables $U$ and $V$, respectively, based on the Gel'€fand-Pinsker binning scheme. Since receiver 2 is interference free and is corrupted by $S_2$, the auxiliary random variable $V$ is designed with regard to only $S_2$. Furthermore,
receiver 1 first decodes $V$, then uses it to cancel the interference $X_2$ and partial state interference, and finally decodes its own message $W_1$ by decoding $U$. Here, since $S_2$ is introduced to $Y_1$ when canceling $X_2$ via $V$, the auxiliary random variable $U$ is designed based on both $S_1$ and $S_2$ to fully cancel the states. Based on such a scheme, we obtain the following achievable region.
\begin{proposition}\label{pps:Z_inner}
	For the state-dependent Z-IC with the states noncausally known at both transmitters, an achievable region consists of rate pairs $(R_1,R_2)$ satisfying:
	\begin{subequations}
		\begin{flalign}
		R_1 \leqslant & I(U;VY_1) -I(S_1S_2;U) \label{eq:zv_inner1}\\
		R_2 \leqslant &  \min\{ I(V;Y_2),I(V;Y_1)\} -I(S_2;V) \label{eq:zv_inner2}
		\end{flalign}
	\end{subequations}
	for some distribution $P_{S_1S_2}P_{U|S_1S_2}P_{X_1|US_1S_2}P_{V|S_2}P_{X_2|VS_2}$
	$P_{Y_1|S_1X_1X_2}P_{Y_2|S_2X_2}$.
\end{proposition}
\begin{proof}
	See Appendix \ref{apx:Z inner}.	
\end{proof}
Following Proposition \ref{pps:Z_inner}, we further simplify the achievable region in the following corollary, which is in a useful form for us to characterize the capacity region for the Gaussian Z-IC. 
\begin{corollary}\label{cor:Z_inner}
	For the state-dependent Z-IC with the states noncausally known at both transmitters, if the following condition
	\vspace{-5mm}	
	\begin{flalign}
	I(V;Y_2)\leqslant I(V;Y_1)\label{zvcond}
	\vspace{-8mm}	
	\end{flalign} 
	is satisfied, then an achievable region consists of rate pairs $(R_1,R_2)$ satisfying:
	\begin{equation}\label{cor:cond1}
		\begin{aligned}
		R_1 \leqslant & I(U;VY_1) -I(S_1S_2;U) \\
		R_2 \leqslant & I(V;Y_2) -I(S_2;V) 
		\end{aligned}
	\end{equation}
	for some distribution $P_{S_1S_2}P_{U|S_1S_2}P_{X_1|US_1S_2}P_{V|S_2}P_{X_2|VS_2}$
	$P_{Y_1|S_1X_1X_2}P_{Y_2|S_2X_2}$.
\end{corollary}
In Corollary \ref{cor:Z_inner}, condition \eqref{zvcond} requires that receiver 1 is more capable in decoding $V$ (and hence $W_2$) than receiver 2, which is likely to be satisfied in the very strong interference regime. 

Following Corollary \ref{cor:Z_inner}, we characterize the channel parameters under which both the states and interference can be fully canceled, and hence the capacity region for the Z-IC is obtained. 

\begin{theorem}\label{thr:VSCapacity}
	For the state-dependent Gaussian Z-IC with states noncausally known at both transmitters, if the channel parameters  $(a,d,P_1,P_2,Q_1^\prime,Q_2)$ satisfy the following condition:
	\begin{equation} \frac{P_1+a^2P_2+d^2Q_2+Q_1^\prime+1}{(d+a\beta)^2Q_2P_2+(P_2+\beta^2Q_2)(P_1+Q_1^\prime+1])}\geqslant \frac{P_2+1}{P_2} \label{eq:verystrongcond}
	\end{equation}
	where $\beta=\frac{P_2}{P_2+1}$, then the capacity region is characterized by \eqref{cap:VeryStrong}.	
	\end {theorem}
\begin{proof}
	Theorem \ref{thr:VSCapacity} follows from Corollary \ref{cor:Z_inner} by setting $U=X_1+\alpha_1S_2+\alpha_2S_1^\prime$, $V=X_2+\beta S_2$, where $X_1$, $X_2$, $ S_1^\prime$ and $ S_2 $ are independent Gaussian variables withe mean zero and variances $P_1$, $P_2$, $ Q_1^\prime $ and $ Q_2 $, respectively. As discussed in the proof of Proposition \ref{pps:Z_inner}, V is first decoded by decoder 1. And then by dirty paper coding, we design $ \alpha_1 $, $\alpha_2$ and $\beta$ for both $Y_2=X_2+S_2+N_2$ and $Y_1^{\prime}=Y_1-aV=X_1+(d-a\beta)S_2+S_1^\prime+N_1$ to fully cancel the states. Thus, the coefficients should satisfy the following conditions:
	\begin{subequations}  
		\begin{flalign}
		\frac{\alpha_1}{d-a\beta}&=\frac{P_1}{P_1+1}\label{eq:zdirty_con1}\\
		\alpha_2&=\frac{P_1}{P_1+1}\label{eq:zdirty_con2}\\
		\beta&=\frac{P_2}{P_2+1},\label{eq:zdirty_con3}
		\end{flalign}
	\end{subequations}
	which further yields
	$\alpha_1$, $\alpha_2$ and $\beta$ that satisfy
	\begin{flalign}
	\alpha_1&=\frac{P_1}{P_1+1}(d-\frac{aP_2}{P_2+1}), \quad \alpha_2=\frac{P_1}{P_1+1},\quad
	\beta=\frac{P_2}{P_2+1}.\ \ \ \ \  \nn
	\end{flalign}
	Substituting the above choice of the auxiliary random variables and the parameters into \eqref{zvcond} in Corollary \ref{cor:Z_inner}, we obtain the condition \eqref{eq:verystrongcond}. Substituting those choices into the condition \eqref{cor:cond1}, we obtain the capacity region characterized by \eqref{cap:VeryStrong}. Since such an achievable region achieves the point-to-point channel capacity for the Z-IC without the state, it can be shown to be the capacity region of the state-dependent Z-IC.
\end{proof}

	\vspace{-3mm}
	\begin{figure}[thb]
		\centering
		\includegraphics[height=2.4in,width=4in]{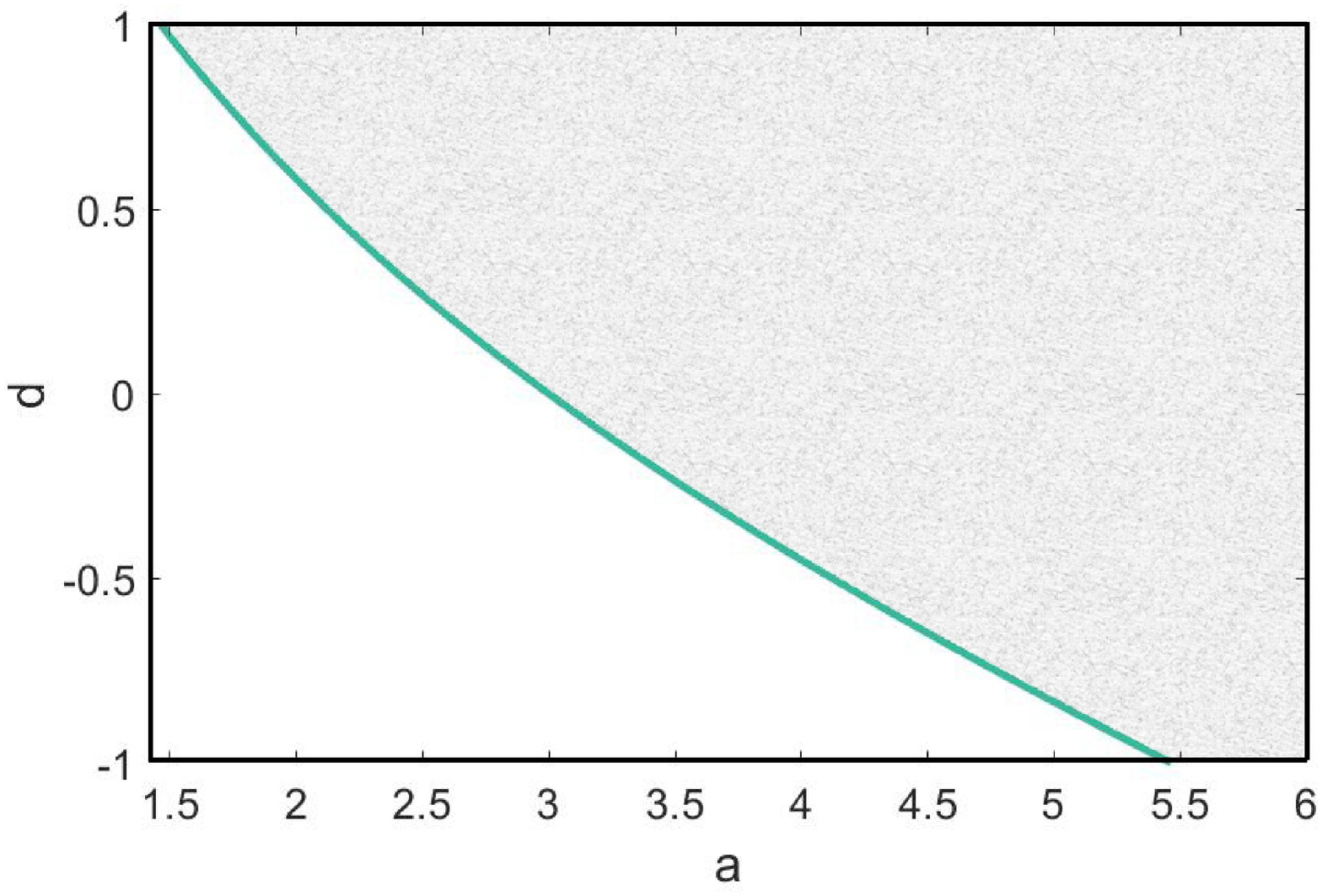}
		\caption{Characterization of channel parameters $(a,d)$ in shaded area under which the state-dependent Gaussian Z-IC achieves the capacity of the corresponding channel without states and interference in very strong interference  regime.}\label{fig:vsac1}
	\end{figure}
	
	Based on Theorem \ref{thr:VSCapacity}, if channel parameters satisfy the condition \eqref{eq:verystrongcond}, we can simultaneously cancel two states and the interference, and the point-to-point capacity of two receivers without state and interference can be achieved. The correlation between the two states captured by $d$ plays a very important role regarding whether the condition can be satisfied.
	In Fig.~\ref{fig:vsac1}, we set $P_1=2$, $P_2=2$, $Q_1=1$ and $Q_2=1$, and plot the range of the parameter pairs $(a,d)$ under which the channel capacity without states and interference can be achieved. These parameters fall in the shaded area above the line. It can be seen that as $d$ becomes larger (i.e., the correlation between the two states increases), the threshold on the parameter $a$ to fully cancel the interference and state becomes smaller. This suggests that more correlated states are easier to cancel together with the interference.
	\begin{figure}[thb]
		\centering
		\includegraphics[height=2.4in,width=4in]{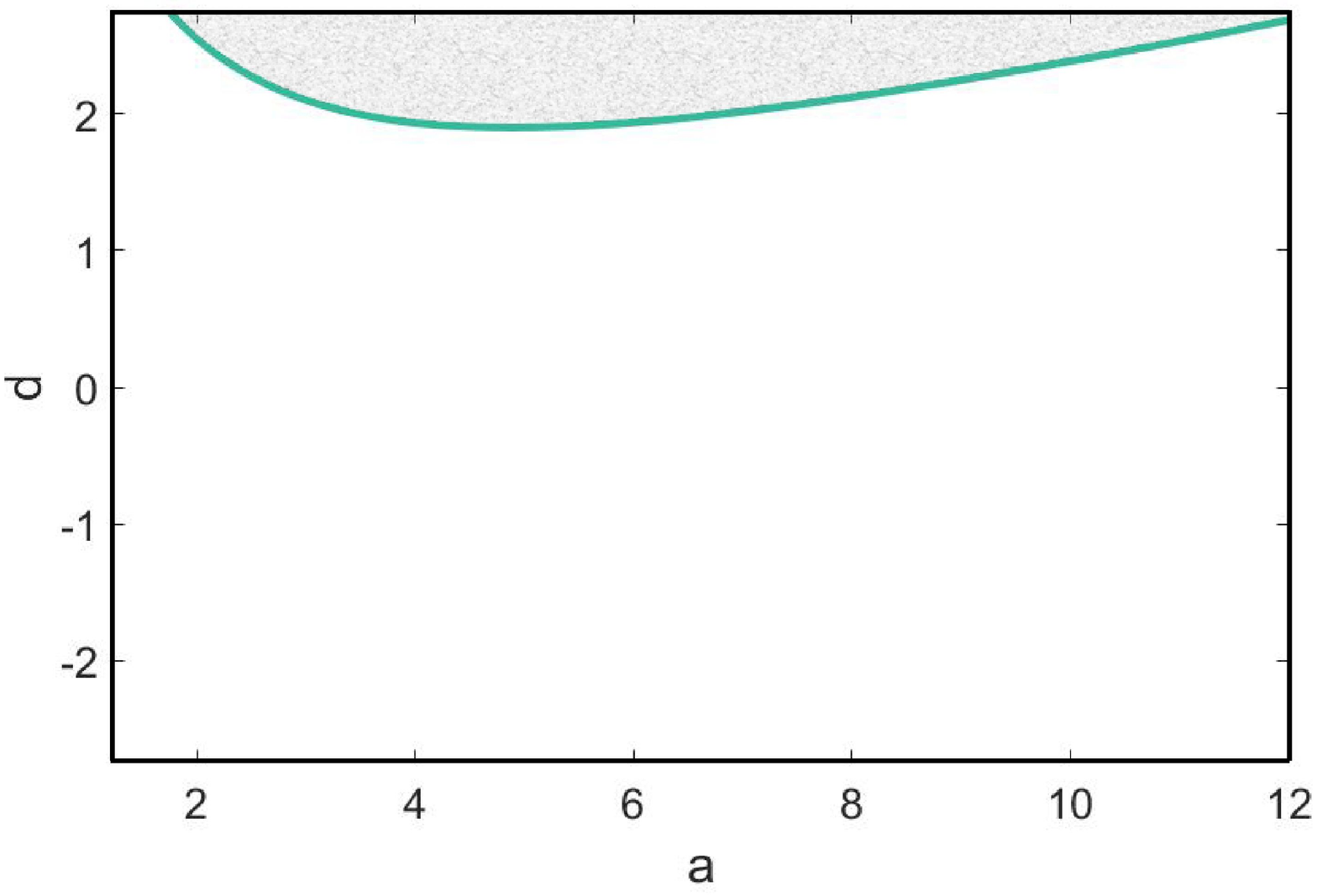}
		\vspace{-2mm}
		\caption{Characterization of channel parameters $(a,d)$ in shaded area under which the state-dependent Gaussian Z-IC achieves the capacity of the corresponding channel without states in very strong interference regime when $Q_2>\frac{1+P_2}{P_2}$.}\label{fig:vsac2}
	\end{figure}

	
	Fig.~\ref{fig:vsac1} agrees with the result of the very strong IC without states in the sense that once $a$ is above a certain threshold (i.e., the interference is strong enough), then the point-to-point channel capacity without interference can be achieved. However, this is not always true for the {\em state-dependent} Z-IC. This can be seen from the condition \eqref{eq:verystrongcond} in Theorem \ref{thr:VSCapacity}. If we let $a$ go to infinity, then the condition \eqref{eq:verystrongcond} becomes $Q_2>\frac{1+P_2}{P_2}$, which is not always satisfied. This is because in the presence of state, $Y_1$ decodes $V$ instead of $X_2$, and the decoding rate is largest if the dirty paper coding design of $V$ (based on $S_2$ at receiver 2) also happens to be the same dirty paper coding design against $S_2$ at receiver 1. Clearly, as $a$ gets too large, $V$ is more deviated from such a favorable design, and hence the decoding rate becomes smaller, which consequently hurts the achievability of the point-to-point capacity for receiver 2. Such a phenomena can be observed in Fig.~\ref{fig:vsac2}, where the constant $a$ cannot be too large to guarantee the achievability of the point-to-point channel capacity. Furthermore, the figure also suggests that further correlated states allow a larger range of $a$ under which the point-to-point channel capacity can be achieved. 	
	

	
	\section{Strong Interference Regime}
In this section, we study the state dependent IC in the strong interference regime, which excludes the very strong interference regime that has been studied in Section \ref{sec:vs_regime}. With the presence of states, these two regimes require separate treatments for state cancellation.

	\subsection{State-Dependent Regular IC}
	
	For the corresponding IC without state, if it is strong but not very strong, then the channel parameters satisfy
	\begin{flalign}
	&a\geqslant 1, \ \ \ \ b\geqslant 1,\\\nn
	&\min\{P_1+a^2P_2+1,b^2P_1+P_2+1\}\leqslant(1+P_1)(1+P_2).
	\end{flalign}
Without loss of generality, we assume that $P_1+a^2P_2+1 \leqslant b^2P_1+P_2+1$. It has been shown in \cite{Sato81} that the capacity region for the strong IC without states contains rate pair $(R_1,R_2)$ satisfying 				
	\begin{flalign}\nn
&	R_1 \leqslant \frac{1}{2}\log(1+P_1) ,\ \ \ \	R_2 \leqslant \frac{1}{2}\log(1+P_2),\\
	&R_1+R_2 \leqslant \frac{1}{2}\log(P_1+a^2P_2+1).\label{eq:strong_cap}
	\end{flalign}

	\begin{figure}[tbh]
	\vspace{-4mm}
	\centering
	\includegraphics[width=4.5in]{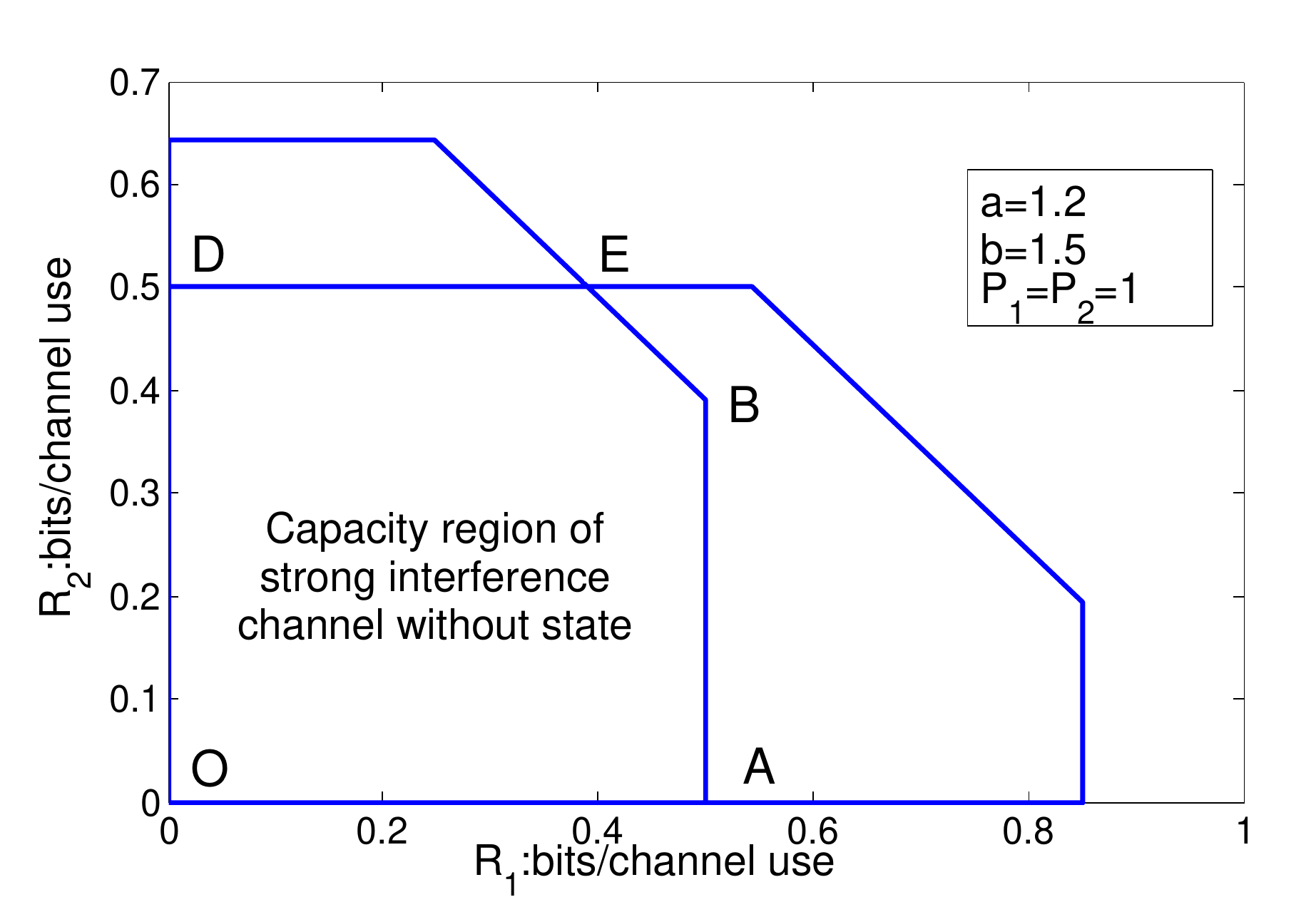}
	\vspace{-2mm}
	\caption{Capacity region of the strong IC without state}\label{fig:sreg}
	\vspace{-1mm}
\end{figure}	
 Such a region is an intersection of the capacity regions of two MACs, which is illustrated as the pentagon O-A-B-E-D-O in Fig. \ref{fig:sreg}. Our goal here is to study whether the points on the sum-rate capacity boundary of the Z-IC without state can be achieved with the presence of state. Such a problem has been studied in \cite{Duan16IT} for the channel with two receivers corrupted by the same but differently scaled state. Here, we generalize such a study to the situation when the two receivers are corrupted by two correlated states. 

Since every point on this line of the sum-rate capacity can be achieved by rate splitting and successive cancellation in the case without state, for the state-dependent channel, we continue to adopt the idea of rate splitting and successive cancellation but using auxiliary random variables to incorporate dirty paper coding to further cancel state successively. More specifically, transmitter 1 splits its message $W_1$ into $W_{11}$ and $W_{12}$, and then encodes them into $U_1$ and $U_2$ respectively based on the Gel'fand-Pinsker binning scheme. Then transmitter 2 encodes its message $W_2$ into \textit{V}, based on the Gel'fand-Pinsker binning scheme.  The auxiliary random variables $U_1$, $U_2$, and $V$ are designed such that decoding of them at receiver 1 successively fully cancels the state corruption of $Y_1$ so that the sum capacity boundary (i.e., the line B-E) can be achieved if only decoding at receiver 1 is considered. Now further incorporating the decoding at receiver 2, if for any point on the line B-E, decoding of $V$ at receiver 2 does not cause further rate constraints, then such a point is achievable for the state-dependent IC.  

	\begin{proposition}\label{pps:strong_inner}
	For the state-dependent IC with states noncausally known at both transmitters, an achievable region consists of rate pairs $(R_1,R_2)$ satisfying:
	\begin{equation}
		\begin{aligned}\label{rate:pps_stronginner}
		R_{1}&\leqslant \min\{I(U_1;Y_1), I(U_1;Y_2)\}\\
		&+\min\{I(U_2;VY_1|U_1), I(U_2;VY_2|U_1)\}-I(U_1U_2;S_1)\\
		R_{2}&\leqslant \min\{I(V;Y_1|U_1), I(V;Y_2|U_1)\}-I(V;S_1)
		\end{aligned}
	\end{equation}
	for some distribution $P_{S_1S_2}P_{V|S_1}P_{X_2|VS_1}P_{U_1|S_1}P_{U_2|S_1U_1}P_{X_1|S_1U_1U_2}P_{Y_1|S_1X_1X_2}P_{Y_2|S_2X_2}$, where $U_1$, $U_2$ and $V$ are auxiliary random variables.
\end{proposition}
\begin{proof}
	See Appendix \ref{apx:strong_inner}.
\end{proof}
\begin{remark}
	This scheme can be generalized through further splitting the messages and changing the orders of decoding the messages at the two receivers. The achievable region can then be obtained by taking the convex hull of the union over all achievable regions corresponding to different schemes above.
\end{remark}

Based on Proposition \ref{pps:strong_inner}, we next characterize partial boundary of the capacity region for the state-dependent Gaussian IC. For the sake of technical convenience, we express the Gaussian model in a different form. In particular, we express $S_2$ as $S_2=cS_1+S_2^\prime$ where $c$ is a constant representing the level of correlation, and $S_1$ is independent from $S_2^\prime$ with $S_2^\prime\sim \mathcal{N}(0, Q_2')$ where $Q_2=c^2Q_1+Q_2^\prime$. Thus, without loss of generality, the channel model can be expressed in the following equivalent form that is more convenient for analysis here.
\begin{subequations}
	\begin{flalign}
	Y_1&=X_1+ aX_2+S_1+N_1\\
	Y_2&=bX_1+X_2+cS_1+S_2^\prime+N_2.
	\end{flalign}
\end{subequations}
We next show that we can design a scheme to achieve the partial boundary of the capacity region for the IC without state. We note that the rate on the sum-capacity boundary can be characterized by
\begin{equation}\label{rate:capIC}
	\begin{aligned}
	&R_1=\frac{1}{2}\log\left(1+\frac{P_1'}{a^2P_2+P_1''+1}\right) + \frac{1}{2}\log\left(1+P_1''\right),\\
	&R_2=\frac{1}{2}\log\left(1+\frac{a^2P_2}{P_1''+1}\right),
\end{aligned}
\end{equation}
for some $P_1^{\prime}$, $P_1^{\prime\prime}\geqslant 0$, and $P_1^{\prime}+P_1^{\prime\prime}\leqslant P_1$.
\begin{theorem} \label{thr:SPointCapa}
	Any rate point in \eqref{rate:capIC} can be achieved by the state-dependent IC if the channel parameters satisfy the following conditions
	\begin{equation}\label{cond:thr_stronginner}
	 \begin{aligned}
	 I(U_1;Y_2)-I(U_1;S_1)&\leqslant\frac{1}{2}\log\left(1+\frac{P_1'}{a^2P_2+P_1''+1}\right)\\
	 I(U_2;VY_2|U_1)-I(U_2;S_1|U_1)&\leqslant \frac{1}{2}\log\left(1+P_1''\right)\\
	 I(V;Y_2|U_1)-I(V;S_1)&\leqslant \frac{1}{2}\log\left(1+\frac{a^2P_2}{P_1''+1}\right),
	 \end{aligned}
	 \end{equation}
	 where the mutual information terms are calculated by setting $U_1=X_1^\prime+\alpha_1S_1$, $U_2=X_1^{\prime\prime}+\alpha_2S_1$, $V=aX_2+\beta S_1$ and $X_1=X_1^\prime+X_1^{\prime\prime}$. Here $X_1^\prime$, $X_1^{\prime\prime}$ and $X_2$ are Gaussian variables with mean zero and variances $P_1^\prime$, $P_1^{\prime\prime}$ and $P_2$, and $\alpha_1$,$\alpha_2$ and $\beta$ are given by
	 \begin{equation} \label{eq:z_condition}
	 	\begin{aligned}
	 	\alpha_1&=\frac{P_1^\prime}{P_1+a^2P_2+1}, \ \ \ \alpha_2=\frac{P_1^{\prime\prime}}{P_1+a^2P_2+1},\\
	 	\beta&=\frac{a^2P_2}{P_1+a^2P_2+1}.
	 	\end{aligned}
	 \end{equation}
\end{theorem}
\begin{proof}
	Theorem \ref{thr:SPointCapa} follows from Proposition \ref{pps:strong_inner} by choosing the auxiliary random variables $U_1$, $U_2$ and $V$ as in the statement of the theorem. In particular, $U_1$ is first decoded by receiver 1, and is designed to cancel the state in $Y_1$ treating all other variables as noise. Then, $V$ is decoded by receiver 1, and is designed to cancel the state in $Y_1^{\prime}=Y_1-U_1=X_1''+aX_2+(c-\alpha_1)S_1+N_1$. Finally, $U_2$ is designed to cancel the state in $Y_1''=Y_1'-V=X_1''+(c-\alpha_1-\beta)S_1+N_1$. In order to satisfy the state cancellation requirements, $\alpha_1$, $\alpha_2$ and $\beta$ should satisfy
	\begin{flalign}
	&\alpha_1=\frac{P_1'}{P_1+a^2P_2+1}, \label{eq:sdirty_con1}\quad \quad\\
	&				\frac{\alpha_2}{1-\alpha_1}=\frac{P_1^{\prime\prime}}{P_1^{\prime\prime}+1}, \label{eq:sdirty_con2}\\
	&				\frac{\beta}{1-\alpha_1}=\frac{a^2P_2}{P_1^{\prime\prime}+a^2P_2+1},\label{eq:sdirty_con3}
	\end{flalign}
	which yields \eqref{eq:z_condition}. Substituting these choices of the random variables and the coefficients into Proposition \ref{pps:strong_inner},  \eqref{rate:pps_stronginner} becomes
	
	\begin{equation}
	\begin{aligned}
	R_{1}&\leqslant \min\left\{ I(U_1;Y_2)-I(U_1;S_1),\frac{1}{2}\log\left(1+\frac{P_1'}{a^2P_2+P_1''+1}\right)\right\}\\
	&+\min\left\{I(U_2;VY_2|U_1)-I(U_2;S_1|U_1), \frac{1}{2}\log\left(1+P_1''\right)\right\}\\
	R_{2}&\leqslant \min\left\{I(V;Y_2|U_1)-I(V;S_1), \frac{1}{2}\log\left(1+\frac{a^2P_2}{P_1''+1}\right)\right\}.
	\end{aligned}
	\end{equation}
Hence, if the condition \eqref{cond:thr_stronginner} is satisfied, the points on the sum capacity boundary \eqref{rate:capIC} can be achieved.
\end{proof} 
\subsection{State-Dependent Z-IC}

In this subsection, we study the state-dependent Z-IC in the strong, but not very strong interference regime, in which the channel parameters satisfy  $1 \leqslant a^2 < 1+P_1$. For the corresponding Z-IC without states, it has been shown that the capacity region contains rate pairs $(R_1,R_2)$ satisfying 
\begin{flalign}
&R_1 +R_2 \leqslant \frac{1}{2} \log(1+P_1 +a^2P_2)\nn\\
&R_1 \leqslant \frac{1}{2} \log{ (1+P_1)},\;\;\;R_2 \leqslant \frac{1}{2} \log{ (1+P_2)} \label{eq:StrongAchivGaussian}
\end{flalign}
which is illustrated as the pentagon O-A-B-E-F in Fig.~\ref{fig:scapa}.
\begin{figure}[thb]
	\vspace{-2mm}
	\centering
	\includegraphics[width=4.5in]{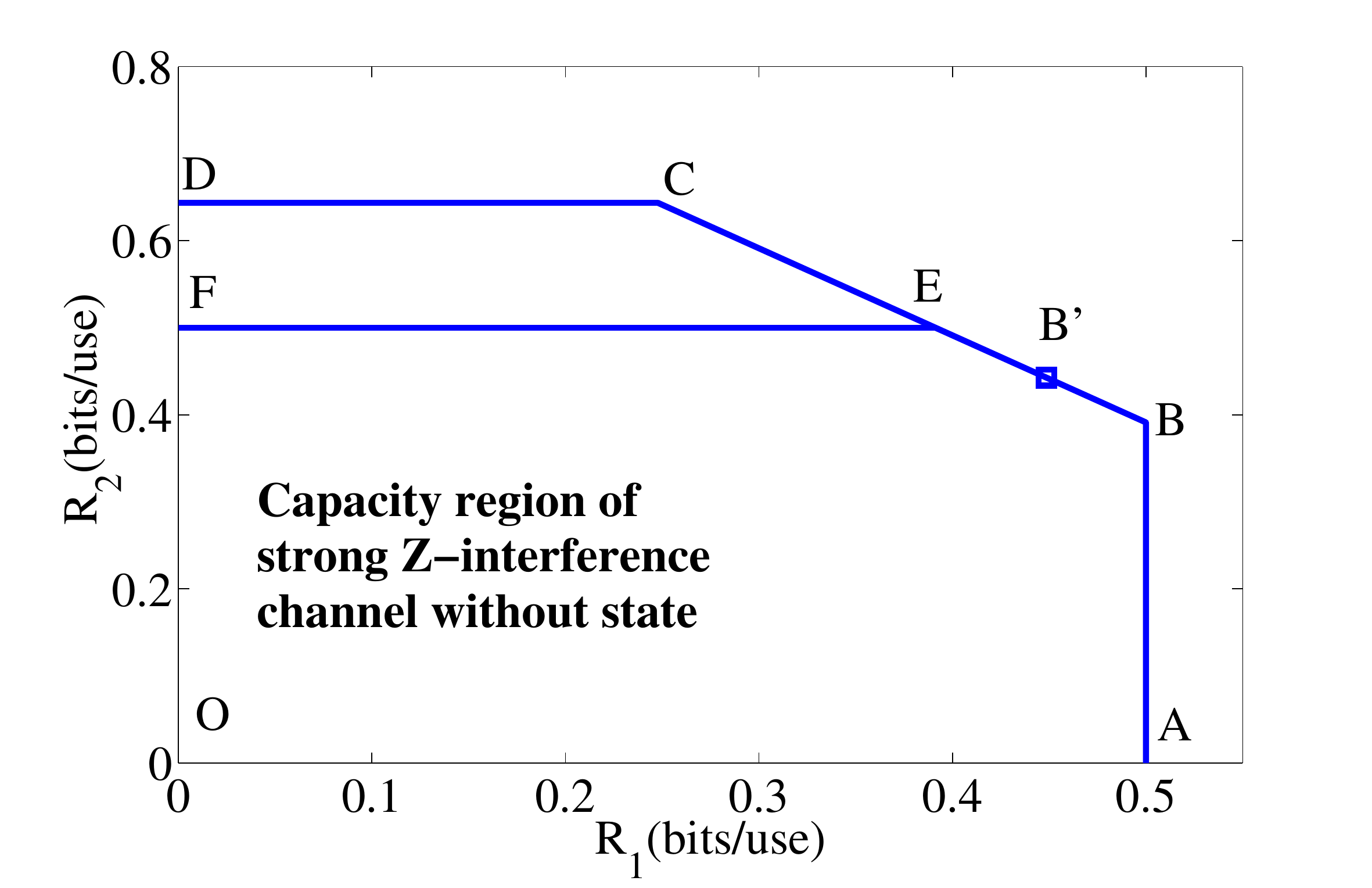}
	\vspace{-2mm}
	\caption{Capacity region of the strong Z-IC without state}\label{fig:scapa}
	\vspace{-1mm}
\end{figure}

Our goal here is to study whether the points on the sum-capacity boundary of the Z-IC {\em without} state (i.e., the line B-E in Fig.~\ref{fig:scapa}) can be achieved for the corresponding {\em state-dependent} Z-IC. We follow the idea for the state-dependent IC to design an achievable scheme, and obtain the following Proposition.
\begin{proposition}\label{pps:strong_Z_inner}
	For the state-dependent Z-IC with states noncausally known at both transmitters, if the following condition is satisfied
	\begin{flalign}
	I(V;U_1Y_1)\leqslant I(V;Y_2),\label{eq:zscond}
	\end{flalign} 
	then an achievable region consists of rate pairs $(R_1,R_2)$ satisfying:
	\begin{subequations}
		\begin{flalign}
		R_1 \leqslant & I(U_1;Y_1)+I(U_2;VY_1|U_1) -I(S_1;U_1U_2)\label{zsas1}\\ 
		R_2 \leqslant &   I(V;U_1Y_1) -I(S_1;V) \label{zsas2}
		\end{flalign}
	\end{subequations}
	for some distribution $P_{S_1S_2}P_{V|S_1}P_{X_2|VS_1}P_{U_1|S_1}P_{U_2|S_1U_1}$ $P_{X_1|S_1U_1U_2}P_{Y_1|S_1X_1X_2}P_{Y_2|S_2X_2}$.
\end{proposition}
The proof of Proposition \ref{pps:strong_Z_inner} is similar to the proof of Proposition \ref{pps:strong_inner}. The only difference lies in that $Y_2$ does not need to decode $U_1$ and $U_2$ because of no interference at receiver 2.

For the Gaussian model,  based on Proposition \ref{pps:strong_Z_inner}, we characterize the condition under which any point on the sum capacity boundary of the strong Z-IC without states (e.g., point $B'$ in Fig. \ref{fig:scapa}) is achievable. Hence, such a point is on the sum capacity boundary of the state-dependent Z-IC.

	\begin{theorem}\label{thr:ZSPointCapa}
		For the state-dependent Gaussian Z-IC with state noncausally known at both transmitters, if the channel parameters  $(a,c,P_1,P_2,Q_1,Q_2^\prime)$ satisfy the following condition:
		\begin{flalign}
		&\frac{a^2P_2(P_2+c^2Q_1+Q_2^\prime+1)}{(ac-\beta)^2Q_1P_2+(a^2P_2+\beta^2Q_1)(Q_2^\prime+1)}  \geqslant 1+\frac{a^2P_2}{P_1^{\prime\prime}+1} \label{eq:ZSCond}
		\end{flalign}
		where $\beta=\frac{a^2P_2}{P_1+a^2P_2+1}$, then the following point (on the line B-E)
		\begin{flalign}
		&R_1=\frac{1}{2}\log\left(1+\frac{P_1'}{a^2P_2+P_1''+1}\right) + \frac{1}{2}\log\left(1+P_1''\right)\nn\\
		&R_2=\frac{1}{2}\log\left(1+\frac{a^2P_2}{P_1''+1}\right) \label{eq:InnerPointsStrong}
		\end{flalign}
		where $P_1'=P_1-P_1''$, is on the sum capacity boundary. 
		\end {theorem}
\begin{proof}
Theorem \ref{thr:ZSPointCapa} follows from Proposition \ref{pps:strong_Z_inner} by choosing the auxiliary random variables $U_1$, $U_2$ and $V$ based on the setting in Theorem \ref{thr:SPointCapa}. Thus, the state at $Y_1$ can be fully canceled. Furthermore, by substituting the auxiliary random variables in Theorem \ref{thr:SPointCapa} into \eqref{eq:zscond}, the condition \eqref{eq:ZSCond} can be obtained, under which the points characterized in \eqref{eq:InnerPointsStrong} can be achieved. 
 \end{proof}

		Theorem \ref{thr:ZSPointCapa} provides the condition of channel parameters under which a certain given point is on the sum capacity boundary of the capacity region. We next characterize a line segment on the sum capacity boundary for a given set of channel parameters.

	\begin{corollary}\label{cor:capacitySegment}
		For the state-dependent Z-IC with state noncausally known at both transmitters, if a point on the line $B-E$ in Fig.~\ref{fig:scapa} is on the sum-capacity boundary for a given set of channel parameters, then the segment between this point and point $B$ on the line $B-E$ is on the sum capacity boundary for the same set of channel parameters.
	\end{corollary}

\vspace{-2mm}		
		\begin{figure}[thb]
			\centering
			\includegraphics[width=3in]{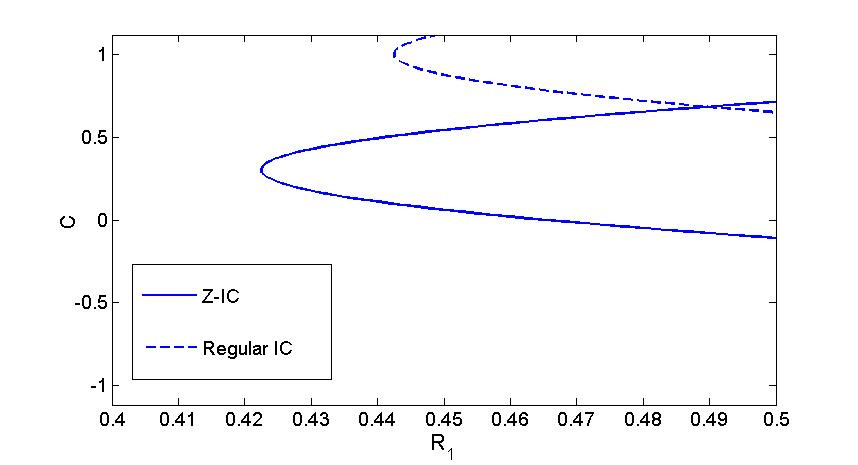}
			\vspace{-2mm}
			\caption{Ranges of $c$ under which points on sum capacity boundary of the strong IC and Z-IC without state can be achieved by the state-dependent IC and Z-IC.}\label{fig:sac}
		\end{figure}


In order to numerically illustrate Theorem \ref{thr:ZSPointCapa}, we first note that each point on the sum-capacity boundary (i.e., the line B-E in  Fig. \ref{fig:scapa}) can be expressed as $(R_1,R_2)=(R_1, \frac{1}{2}\log(P_1+a^2P_2+1)-R_1)$. We now set $P_1=2$, $P_2=0.7$, $Q_1=0.4$, $Q_2=0.5$ and $a=1.2$, and hence $R_1 \in [ \frac{1}{2}\log(1.23), 0.5]$ parameterizes all points from point E to point B in Fig. \ref{fig:scapa}. In Fig.~\ref{fig:sac}, we plot the ranges of $c$ under which points, parameterized by $R_1$ on the sum capacity boundary of the strong Z-IC without state, can be achieved by the state-dependent Z-IC following Theorem \ref{thr:ZSPointCapa}. It can be seen that as correlation between the two states (represented by $c$) increases, initially more points on the sum capacity boundary are achieved and then  less points are achieved as $c$ is above a certain threshold. Thus, higher correlation does not guarantee more capability of achieving the sum capacity boundary. This is because in our scheme both $U_i$ and $V$ are specially designed for $Y_1$ based on dirty paper coding. At receiver 2, such design of $V$ initially approximates better the dirty paper coding design for $Y_2$ as $c$ becomes large, but then becomes worse as $c$ continues to increase. Hence decoding of $V$ at receiver 2 initially gets better and then becomes less capable, which consequently determines variation of achievability of the sum capacity boundary. 

Fig. \ref{fig:sac} also plots the same parameter range for the state-dependent IC as characterized by Theorem \ref{thr:SPointCapa}. It is clear that the state-dependent IC achieves a smaller line segment on the sum-capacity (i.e., smaller range of $R_1$). This is reasonable, because Theorem \ref{thr:SPointCapa} for the IC requires more conditions than Theorem \ref{thr:ZSPointCapa} for the Z-IC. Fig. \ref{fig:sac} also demonstrates that large value of c(i.e., higher correlation between the states) is required for the IC to achieve the sum capacity than the Z-IC. This is because the dirty paper coding is designed with respect to receiver 1. High correlation between states helps such design to be more effective to cancel state at receiver 2 as well.

\section{Weak Interference Regime}\label{sec:weak}
		In this section, we study the state-dependent IC and Z-IC in the weak interference regime. The channel parameters for the IC in this regime satisfy  $|a(1+b^2P_1)|+|b(1+a^2P_2)|\leqslant1$, which reduces to $a\leqslant1$ for the Z-IC. It has been shown in \cite{Shang09,Anna09,Mota09}, for the weak IC without state and in \cite{Sason04} for the weak Z-IC that the sum capacity can be achieved by treating interference as noise. It was further shown in \cite{Duan16IT} that for the IC and Z-IC with the same but differently scaled state at two receivers, independent dirty paper coding at the two transmitters to cancel the states and treating interference as noise achieve the same sum capacity. We here observe that such a scheme is also achievable with the presence of two correlated states, which thus yields the following Corollary.
			
		\begin{corollary}(A direct result following \cite{Duan16IT})\label{thr:weak_IC_ZIC}
			For the state-dependent IC with states noncausally known at both transmitters, if $|a(1+b^2P_1)|+|b(1+a^2P_2)|\leqslant1$, then the sum capacity is given by
			\begin{flalign}\nn
			C_{sum} = \frac{1}{2}\log\left(1+\frac{P_1}{a^2P_2+1}\right)+ \frac{1}{2}\log\left(1+\frac{P_2}{b^2P_1+1}\right).
			\end{flalign}
			For the state-dependent Z-IC with states noncausally known at both transmitters, if $a^2\leqslant 1$, then the sum capacity is given by
			\begin{flalign}\nn
			C_{sum} = \frac{1}{2}\log\left(1+\frac{P_1}{a^2P_2+1}\right)+ \frac{1}{2}\log\left(1+P_2\right).
			\end{flalign}
		\end{corollary}
	
It can be seen that the sum capacity achieving scheme does not depend on the correlation of the states, and hence, in the weak interference regime, the sum capacity is not affected by the correlation of the states.

\section{Conclusion}\label{sec:Conclusion}
		
		In this paper, we studied the state-dependent Gaussian IC and Z-IC with receivers being corrupted by two {\em correlated} states which are noncausally known at transmitters. The correlated states can be reduced to two extreme cases: two independent states and one differently scaled state. We characterized the conditions on the channel parameters under which state-dependent IC and Z-IC achieve the capacity region or the sum capacity of the corresponding channel without state. Our result suggests that more correlated states tend to make it easier to fully cancel the states. Our comparison between the IC and the Z-IC suggests that the IC benefits more if the correlation between the states increases. We anticipates that the state cancellation schemes we develop here can be useful for studying other state-dependent models.

\vspace{1cm}

\appendix

\noindent {\Large \textbf{Appendix}}

\vspace{-3mm}
\section{Proof of Proposition \ref{pps:IC inner}}\label{apx:IC inner}
We use random codes and fix the following joint distribution:
$$P_{S_1S_2UVX_1X_2Y_1Y_2}=P_{S_1S_2}P_{U|S_1S_2}P_{X_1|US_1S_2}P_{V|S_1S_2}P_{X_2|VS_1S_2}P_{Y_1Y_2|S_1S_2X_1X_2}.$$

\begin{enumerate}
	\item Codebook Generation:
	\begin{itemize}
		\item Generate $2^{n(R_1+\tR_1)}$ codewords $U^n(w_1,l_1)$ with i.i.d.\ components based on $P_U$. Index these codewords by $w_1=1, \ldots, 2^{nR_1}, l_1 = 1, \ldots, 2^{n\tR_1}$.
		\item Generate $2^{n(R_2+\tR_2)}$ codewords $V^n(w_2,l_2)$ with i.i.d.\ components based on $P_V$. Index these codewords by $w_2=1, \ldots, 2^{nR_2}, l_2 = 1, \ldots, 2^{n\tR_2}$.
	\end{itemize}
	\item Encoding:
	\begin{itemize}	
		\item Transmitter 1: Given $(s^n_1,s^n_2)$ and $w_1$, choose a $u^n(w_1,\tl_1)$ such that $$(u^n(w_1,\tl_1),s_1^n,s_2^n) \in T^n_\epsilon(P_{S_1S_2U}).$$ Otherwise, set $\tl_1=1$. It can be shown that for large $n$, such $u^n$ exists with high probability if 
		\begin{equation}
		\tR_1>I(U;S_1S_2). \label{eq:pps1}
		\end{equation}
		Then generate $x^n_1$ with i.i.d. component based on $P_{X_1|US_1S_2}$ for transmission.

		\item Transmitter 2: Given $(s^n_1,s^n_2)$ and $w_2$, choose a $v^n(w_2,\tl_2)$ such that $$(v^n(w_2,\tl_2),s_1^n,s_2^n) \in T^n_\epsilon(P_{S_1S_2V}).$$
		Otherwise, set $\tl_2=1$. It can be shown that for large $n$, such $v^n$ exists with high probability if 	
		\begin{equation}
		\tR_2>I(V;S_1S_2).
		\end{equation}
		Then generate $x^n_2$ with i.i.d. components based on $P_{X_2|US_1S_2}$ for transmission.
	\end{itemize}

	\item Decoding:
	\begin{itemize}
		\item Decoder 1: Given $y^n_1$, find $(\hw_2, \hl_2)$ such that $$(v^n(w_2,\hl_2),y^n_1) \in T^n_\epsilon(P_{VY_1}).$$  If no or more than one such pair $(\hw_2, \hl_2)$ can be found, declare an error. It is easy to show that for sufficiently large $n$, we can correctly find such a pair with high probability if 
		\begin{equation}
		R_2+\tR_2\leqslant I(V;Y_1).
		\end{equation}
		After decoding $v^n$, find a unique pair $(\hw_1, \hl_1)$ such that $$(u^n(\hw_1, \hl_1),v^n(w_2,\hl_2),y^n_1) \in T^n_\epsilon(P_{VUY_1}).$$
		If no or more than one such pair $(\hw_2, \hl_2)$ can be found, declare an error.  It is easy to show that for sufficiently large $n$, we can correctly find such a pair with high probability if 
		\begin{equation}
		R_1+\tR_1\leqslant I(U;VY_1).
		\end{equation}
		
		\item Decoder 2: Given $y^n_2$, find $(\hw_1, \hl_1)$ such that $$(v^n(w_1,\hl_1),y^n_2) \in T^n_\epsilon(P_{UY_2}).$$  If no or more than one such pair $(\hw_1, \hl_1)$ can be found, declare an error. It is easy to show that for sufficiently large $n$, we can correctly find such a pair with high probability if 
		\begin{equation}
		R_1+\tR_1\leqslant I(U;Y_2).
		\end{equation}
		After decoding $u^n$, find a unique pair $(\hw_2, \hl_2)$ such that $$(v^n(\hw_2, \hl_2),u^n(w_1,\hl_1),y^n_2) \in T^n_\epsilon(P_{VUY_2}).$$
		If no or more than one such pair can be found, declare an error.  It is easy to show that for sufficiently large $n$, we can correctly find such a pair with high probability if 
		\begin{equation}
		R_2+\tR_2\leqslant I(V;UY_2).\label{eq:pps2}
		\end{equation}
	\end{itemize}
\end{enumerate}
Proposition \ref{pps:IC inner} is thus proved by combining \eqref{eq:pps1}-\eqref{eq:pps2}.

\section{Proof of Proposition \ref{pps:Z_inner}}\label{apx:Z inner}
We use random codes and fix the following joint distribution:
$$P_{S_1S_2UVX_1X_2Y_1Y_2}=P_{S_1S_2}P_{U|S_1S_2}P_{X_1|US_1S_2}P_{V|S_1S_2}P_{X_2|VS_1S_2}P_{Y_1Y_2|S_1S_2X_1X_2}.$$

\begin{enumerate}
	\item Codebook Generation:
	\begin{itemize}
		\item Generate $2^{n(R_1+\tR_1)}$ codewords $U^n(w_1,l_1)$ with i.i.d.\ components based on $P_U$. Index these codewords by $w_1=1, \ldots, 2^{nR_1}, l_1 = 1, \ldots, 2^{n\tR_1}$.
		\item Generate $2^{n(R_2+\tR_2)}$ codewords $V^n(w_2,l_2)$ with i.i.d.\ components based on $P_V$. Index these codewords by $w_2=1, \ldots, 2^{nR_2}, l_2 = 1, \ldots, 2^{n\tR_2}$.
	\end{itemize}
	\item Encoding:
	\begin{itemize}
		
		\item Transmitter 1: Given $(s^n_1,s^n_2)$ and $w_1$, choose a $u^n(w_1,\tl_1)$ such that $$(u^n(w_1,\tl_1),s_1^n,s_2^n) \in T^n_\epsilon(P_{S_1S_2U}).$$
		Otherwise, set $\tl_1=1$. It can be shown that for large $n$, such $u^n$ exists with high probability if 
		\begin{equation}
		\tR_1>I(U;S_1S_2). \label{eq:pps2-1}
		\end{equation}
		Then generate $x^n_1$ with i.i.d. component based on $P_{X_1|US_1S_2}$ for transmission.
		
		\item Transmitter 2: Given $(s^n_1,s^n_2)$ and $w_2$, choose a $v^n(w_2,\tl_2)$ such that $$(v^n(w_2,\tl_2),s_1^n,s_2^n) \in T^n_\epsilon(P_{S_1S_2V}).$$
		Otherwise, set $\tl_2=1$. It can be shown that for large $n$, such $v^n$ exists with high probability if 	
		\begin{equation}
		\tR_2>I(V;S_1S_2).
		\end{equation}
		Then generate $x^n_2$ with i.i.d. component based on $P_{X_2|VS_1S_2}$ for transmission.
	\end{itemize}
	
	\item Decoding:
	\begin{itemize}
		\item Decoder 1: Given $y^n_1$, find $(\hw_2, \hl_2)$ such that $$(v^n(\hw_2,\hl_2),y^n_1) \in T^n_\epsilon(P_{VY_1}).$$  If no or more than one such pair $(\hw_2, \hl_2)$ can be found, declare an error. It is easy to show that for sufficiently large $n$, we can correctly find such a pair with high probability if 
		\begin{equation}
		R_2+\tR_2\leqslant I(V;Y_1).
		\end{equation}
		After decoding $v^n$, find a unique pair $(\hw_1, \hl_1)$ such that $$(u^n(\hw_1, \hl_1),v^n(w_2,\hl_2),y^n_1) \in T^n_\epsilon(P_{VUY_1}).$$
		If no or more than one such pair can be found, declare an error.  It is easy to show that for sufficiently large $n$, we can correctly find such a pair with high probability if 
		\begin{equation}
		R_1+\tR_1\leqslant I(U;VY_1).
		\end{equation}
		
		\item Decoder 2: Given $y^n_2$, find $(\hw_2, \hl_2)$ such that $$(v^n(w_2,\hl_1),y^n_2) \in T^n_\epsilon(P_{UY_2}).$$ If no or more than one such pair $(\hw_2, \hl_2)$ can be found, declare an error. It is easy to show that for sufficiently large $n$, we can correctly find such a pair with high probability if 
		\begin{equation}
		R_2+\tR_2\leqslant I(V;Y_2).\label{eq:pps2-2}
		\end{equation}
	\end{itemize}
\end{enumerate}
Proposition \ref{pps:Z_inner} is thus proved by combining \eqref{eq:pps2-1}-\eqref{eq:pps2-2}.

\section{Proof of Proposition \ref{pps:strong_inner}}\label{apx:strong_inner}
We use random codes and fix the following joint distribution:
$$P_{S_1S_2U_1U_2VX_1X_2Y_1Y_2}=P_{S_1S_2}P_{V|S_1}P_{X_2|VS_1}P_{U_1|S_1}P_{U_2|S_1U_1}P_{X_1|U_1U_2S_1}P_{Y_1|S_1X_1X_2}P_{Y_2|S_2X_2}.$$

\begin{enumerate}
	\item Codebook Generation:
	\begin{itemize}
		\item Generate $2^{n(R_{11}+\tR_{11})}$ codewords $U_1^n(w_{11},l_{11})$ with i.i.d.\ components based on $P_{U_1}$. Index these codewords by $w_{11}=1, \ldots, 2^{nR_{11}}, l_{11} = 1, \ldots, 2^{n\tR_{11}}$.
		\item For each $u_1^n(w_{11},l_{11})$, generate $2^{n(R_{12}+\tR_{12})}$ codewords $U_2^n(w_{11},l_{11},w_{12},l_{12})$ with i.i.d.\ components based on $P_{U_2|U_1}$. Index these codewords by $w_{12}=1, \ldots, 2^{nR_{12}}, l_{12} = 1, \ldots, 2^{n\tR_{12}}$.
		\item Generate $2^{n(R_2+\tR_2)}$ codewords $V^n(w_2,l_2)$ with i.i.d.\ components based on $P_V$. Index these codewords by $w_2=1, \ldots, 2^{nR_2}, v = 1, \ldots, 2^{n\tR_2}$.
	\end{itemize}
	\item Encoding:
	\begin{itemize}
		\item Transmitter 1: Given $s^n_1$ and $w_{11}$, choose a $u_1^n(w_{11},\tl_{11})$ such that $$(u^n(w_{11},\tl_{11}),s_1^n) \in T^n_\epsilon(P_{S_1U_{11}}).$$
		Otherwise, set $\tl_{11}=1$. It can be shown that for large $n$, such $u_1^n$ exists with high probability if 
		\begin{equation}
		\tR_{11}>I(U_1;S_1). \label{eq:pps4-1}
		\end{equation}
		Given $w_{11}$, $\tl_{11}$, $w_{12}$ and $s_1^n$, choose a $u_2^n(w_{11},\tl_{11},w_{12},\tl_{12})$ such that 
		$$(u_1^n(w_{11},\tl_{11}),u_2^n(w_{11},\tl_{11},w_{12},\tl_{12}),s_1^n) \in T^n_\epsilon(P_{S_1U_1U_2}).$$
		Otherwise, set $\tl_{12}=1$. It can be shown that for large $n$, such $u_2^n$ exists with high probability if
		\begin{equation}
		\tR_{12}>I(U_2;S_1|U_1). \label{eq:pps4-2}
		\end{equation}
		Given $u_1^n(w_{11},\tl_{11})$, $u_2^n(w_{11},\tl_{11},w_{12},\tl_{12})$ and $s_1^n$, generate $x_1^n$ with i.i.d. components based on $P_{X_1|S_1U_1U_2}$.
		
		\item Transmitter 2: Given $s^n_1$ and $w_2$, choose a $v^n(w_2,\tl_2)$ such that $$(v^n(w_2,\tl_2),s_1^n) \in T^n_\epsilon(P_{S_1V}).$$
		Otherwise, set $\tl_2=1$. It can be shown that for large $n$, such $v^n$ exists with high probability if 	  
		\begin{equation}
		\tR_2>I(V;S_1).
		\end{equation}
		Then generate $x^n_2$ with i.i.d. component based on $P_{X_2|VS_1}$ for transmission.
	\end{itemize}
	
	\item Decoding:
	\begin{itemize}
		\item Decoder 1: Given $y^n_1$, find $(\hw_{11}, \hl_{11})$ such that $$(u_1^n(\hw_{11}, \hl_{11}),y^n_1) \in T^n_\epsilon(P_{U_1Y_1}).$$  If no or more than one such a pair $(\hw_{11}, \hl_{11})$ can be found, declare an error. It is easy to show that for sufficiently large $n$, we can correctly find such a pair with high probability if 
		\begin{equation}
		R_{11}+\tR_{11}\leqslant I(U_1;Y_1).
		\end{equation}
		After decoding $u_1^n$, find a unique pair $(\hw_{2}, \hl_{2})$ such that $$(u_1^n(\hw_{11}, \hl_{11}),v^n(w_2,\hl_2),y^n_1) \in T^n_\epsilon(P_{VU_1Y_1}).$$
		If no or more than one such pair can be found, declare an error.  It is easy to show that for sufficiently large $n$, we can correctly find such a pair with high probability if 
		\begin{equation}
		R_2+\tR_2\leqslant I(V;Y_1|U_1).
		\end{equation}
		
		After successively decoding $v^n$, find a unique tuple $(w_{11},\tl_{11},w_{12},\tl_{12})$ such that
		$$((u_1^n(\hw_{11}, \hl_{11}),v^n(w_2,\hl_2),u_2^n(w_{11},\tl_{11},w_{12},\tl_{12}),y^n_1) \in T^n_\epsilon(P_{VU_1U_2Y_1})).$$
		If no or more than one such pair can be found, declare an error.  It is easy to show that for sufficiently large $n$, we can correctly find such a pair with high probability if 
		\begin{equation}
		R_{12}+\tR_{12}\leqslant I(U_2;VY_1|U_1).
		\end{equation}

		\item Decoder 2:  Given $y^n_2$, find $(\hw_{11}, \hl_{11})$ such that $$(u_1^n(\hw_{11}, \hl_{11}),y^n_2) \in T^n_\epsilon(P_{U_1Y_1}).$$  If no or more than one such pair $(\hw_{11}, \hl_{11})$ can be found, declare an error. It is easy to show that for sufficiently large $n$, we can correctly find such a pair with high probability if 
		\begin{equation}
		R_{11}+\tR_{11}\leqslant I(U_1;Y_2).
		\end{equation}
		After decoding $u_1^n$, find a unique pair $(\hw_{2}, \hl_{2})$ such that $$(u_1^n(\hw_{11}, \hl_{11}),v^n(w_2,\hl_2),y^n_2) \in T^n_\epsilon(P_{VU_1Y_2}).$$
		If no or more than one such pair can be found, declare an error.  It is easy to show that for sufficiently large $n$, we can correctly find such a pair with high probability if 
		\begin{equation}
		R_2+\tR_2\leqslant I(V;Y_2|U_1).
		\end{equation}
		
		After successively decoding $v^n$, find a unique tuple $(w_{11},\tl_{11},w_{12},\tl_{12})$ such that
		$$((u_1^n(\hw_{11}, \hl_{11}),v^n(w_2,\hl_2),u_2^n(w_{11},\tl_{11},w_{12},\tl_{12}),y^n_1) \in T^n_\epsilon(P_{VU_1U_2Y_2})).$$
		If no or more than one such pair can be found, declare an error.  It is easy to show that for sufficiently large $n$, we can correctly find such a pair with high probability if 
		\begin{equation}
		R_{12}+\tR_{12}\leqslant I(U_2;VY_2|U_1).
		\end{equation}
	\end{itemize}
	The corresponding achievable region is thus characterized by
	\begin{flalign}
		R_{11}&\leqslant \min\{I(U_1;Y_1), I(U_1;Y_2)\}-I(U_1;S_1)\label{eq:pps4-3}\\
		R_{12}&\leqslant \min\{I(U_2;VY_1|U_1), I(U_2;VY_2|U_1)\}-I(U_2;S_1|U_1)\label{eq:pps4-4}\\
		R_{2}&\leqslant \min\{I(V;Y_1|U_1), I(V;Y_2|U_1)\}-I(V;S_1).\label{eq:pps4-5}
	\end{flalign}
\end{enumerate}
Proposition \ref{pps:strong_inner} is completed by combining \eqref{eq:pps4-3}-\eqref{eq:pps4-5} and $R_1=R_{11}+R_{12}$.

\bibliographystyle{unsrt}

\end{document}